\documentclass{article}

\usepackage[affil-it]{authblk}
\usepackage{amsfonts}
\usepackage{amsmath,amsthm,amssymb,dsfont}
\usepackage{enumerate}
\usepackage[english]{babel}
\usepackage{graphicx}	
\usepackage[margin=2.9cm]{geometry}
\usepackage{url}
\usepackage{todonotes}
\usepackage{bbm}

\usepackage{hyperref}
\hypersetup{colorlinks=true,citecolor=blue,linkcolor=blue,filecolor=blue,urlcolor=blue,breaklinks=true}

\usepackage{tikz}
\usetikzlibrary{chains}
\usetikzlibrary{fit}
\usepackage{pgflibraryarrows}		
\usepackage{pgflibrarysnakes}		
\usepackage{xcolor}
\usepackage{epsfig}
\usetikzlibrary{shapes.symbols,patterns} 
\usepackage{pgfplots}

\theoremstyle{plain}
\newtheorem{theorem}{Theorem}[section]
\newtheorem{lemma}[theorem]{Lemma}

\newtheorem{corollary}[theorem]{Corollary}
\newtheorem{proposition}[theorem]{Proposition}

\theoremstyle{definition}
\newtheorem{definition}[theorem]{Definition}
\newtheorem{remark}[theorem]{Remark}
\newtheorem{example}[theorem]{Example}

\newtheorem*{property}{Required properties of the $\Delta$-function}

\newcommand*{\cB}{\mathcal{B}}

\newcommand*{\cE}{\mathcal{E}}
\newcommand*{\cF}{\mathcal{F}}

\newcommand*{\cI}{\mathcal{I}}

\newcommand*{\cN}{\mathcal{N}}
\newcommand*{\cM}{\mathcal{M}}
\newcommand*{\cP}{\mathcal{P}}

\newcommand*{\cR}{\mathcal{R}}
\newcommand*{\cS}{\mathcal{S}}
\newcommand*{\cT}{\mathcal{T}}
\newcommand*{\cU}{\mathcal{U}}

\newcommand*{\MM}{\mathbb{M}}
\newcommand*{\N}{\mathbb{N}}
\newcommand*{\R}{\mathbb{R}}

\newcommand*{\eps}{\varepsilon}
\newcommand*{\rank}{\mathrm{rank}}

\newcommand*{\id}{\mathrm{id}}
\newcommand*{\tr}{\mathrm{tr}}
\newcommand*{\ket}[1]{| #1 \rangle}
\newcommand*{\bra}[1]{\langle #1 |}
\newcommand*{\spr}[2]{\langle #1 | #2 \rangle}
\newcommand{\proj}[1]{|#1\rangle\!\langle #1|}

\newcommand*{\Pos}{\mathrm{Pos}}
\newcommand*{\D}{\mathrm{D}}

\newcommand*{\TPCP}{\mathrm{TPCP}}

\newcommand*{\MD}{D_{\mathbb{M}}}
\newcommand*{\dd}{|\hspace{-0.4mm}|}

\newcommand*{\ci}{\mathrm{i}} 

\newcommand{\norm}[1]{\left\lVert#1\right\rVert}

\allowdisplaybreaks    

\begin{document}

\title{Universal recovery map for approximate Markov chains}

\author[1]{David Sutter}
\author[2,3]{Omar Fawzi}
\author[1]{Renato Renner}
\affil[1]{Institute for Theoretical Physics, ETH Zurich, Switzerland}
\affil[2]{Department of Computing and Mathematical Sciences, Caltech, USA}
\affil[3]{LIP\hspace{0.3mm}\footnote{UMR 5668 LIP - ENS Lyon - CNRS - UCBL - INRIA, Universit\'e de Lyon}, ENS de Lyon, France}
\date{}

\maketitle

\begin{abstract}
A central question in quantum information theory is to determine how well lost information can be reconstructed. Crucially, the corresponding recovery operation should perform well without knowing the information to be reconstructed. In this work, we show that the \emph{quantum conditional mutual information} measures the performance of such recovery operations. More precisely, we prove that the conditional mutual information $I(A\! :\! C | B)$ of a tripartite quantum state $\rho_{ABC}$ can be bounded from below by its distance to the closest recovered state $\cR_{B \to BC}(\rho_{AB})$, where the $C$-part is reconstructed from the $B$-part only and the recovery map $\cR_{B \to BC}$ merely depends on $\rho_{BC}$. One particular application of this result implies the equivalence between two different approaches to define topological order in quantum systems.

\end{abstract}

\section{Introduction}
A state $\rho_{ABC}$ on a tripartite quantum system $A \otimes B \otimes C$ forms a \emph{(quantum) Markov chain} if it can be recovered from its marginal $\rho_{AB}$ on $A\otimes B$ by a quantum operation $\cR_{B \to BC}$ from $B$ to $B\otimes C$, i.e.,
\begin{align}  \label{eq_markov}
\rho_{ABC} = \cR_{B \to BC} (\rho_{AB})\ .
\end{align}
An equivalent characterization of $\rho_{ABC}$ being a Markov chain is that the \emph{conditional mutual information} $I(A:C|B)_{\rho}:=H(AB)_{\rho}+H(BC)_{\rho}-H(B)_{\rho}-H(ABC)_{\rho}$ is zero~\cite{LieRus73,Pet86,Pet03} where $H(A)_{\rho} :=-\tr( \rho_A \log_2 \rho_A)$ is the von Neumann entropy. The structure of these states has been studied in various works. In particular, it has been shown that $A$ and $C$ can be viewed as independent conditioned on $B$, for a meaningful notion of conditioning~\cite{HJPW04}. Very recently it has been shown that Markov states can be alternatively characterized by having a generalized R\'enyi conditional mutual  information that vanishes~\cite{DW15}.

A natural question that is relevant for applications is whether the above statements are robust (see~\cite{Kim13} for an example and~\cite{FR14} for a discussion). Specifically, one would like to have a characterization of the states that have a small (but not necessarily vanishing) conditional mutual information, i.e., ${I(A:C|B)} \leq \eps$ for $\eps > 0$.
First results revealed that such states can have a large distance to Markov chains that is independent of $\eps$~\cite{CSW12,ILW08}, which has been taken as an indication that their characterization may be difficult.
However, it has subsequently been realized that a more appropriate measure instead of the distance to a (perfect) Markov chain is to consider how well~\eqref{eq_markov} is satisfied~\cite{WL12,Kim13,Zha12,BSW14}. This motivated the definition of \emph{approximate Markov chains} as states where~\eqref{eq_markov} approximately holds.

In recent work~\cite{FR14}, it has been shown that the set of approximate Markov chains indeed coincides with the set of states with small conditional mutual information. In particular, the distance between the two terms in~\eqref{eq_markov}, which may be measured in terms of their fidelity $F$, is bounded by the conditional mutual information.\footnote{The fidelity of $\rho$ and $\sigma$ is defined as $F(\rho,\sigma):=\norm{\sqrt{\rho}\sqrt{\sigma}}_1$.}
More precisely, for any state $\rho_{ABC}$ there exists a trace-preserving completely positive map $\cR_{B \to BC}$ (the \emph{recovery map}) such that
\begin{align} \label{eq_FR}
I(A:C|B)_{\rho} \geq - 2 \log_2 F \bigl(\rho_{ABC},\cR_{B \to BC}(\rho_{AB})\bigr) \ .
\end{align}
Furthermore, a converse inequality of the form $I(A:C|B)^2_{\rho} \leq -c^2 \log_2 F(\rho_{ABC},\cR_{B \to BC}(\rho_{AB}))$, where $c$ depends logarithmically on the dimension of $A$ can be shown to hold~\cite{BSW14,FR14}.

We also note that the fidelity term in~\eqref{eq_FR}, maximized over all recovery maps, i.e., 
\begin{align}
F(A;C|B)_{\rho}:= \sup_{\cR_{B \to BC}} F\bigl(\rho_{ABC},\cR_{B \to BC}(\rho_{AB})\bigr) \ 
\end{align}
is called \emph{fidelity of recovery},\footnote{We note that if $A$, $B$, and $C$ are finite-dimensional Hilbert spaces the supremum is achieved, since the set of recovery maps is compact (see Remark~\ref{rmk_TPCPcompact}) and the fidelity is continuous in the input state (see Lemma~B.9 in~\cite{FR14}).} and has been introduced and studied in~\cite{SW14,TB15}.
With this quantity the main result of~\cite{FR14} can be written as
\begin{align}
I(A:C|B)_{\rho} \geq - 2 \log_2 F(A;C|B)_{\rho} \ .
\end{align}
The fidelity of recovery has several natural properties, e.g., it is monotonous under local operations on $A$ and $C$, and it is multiplicative~\cite{TB15}.

The result of~\cite{FR14} has been extended in various ways. Based on quantum state redistribution protocols, it has been shown in~\cite{BHOS14} that~\eqref{eq_FR} still holds if the fidelity term is replaced by the \emph{measured relative entropy}\footnote{The measured relative entropy is defined in Appendix~\ref{ap_measRelEnt} (Definition~\ref{def_MeasRelEnt}).} $\MD(\cdot,\cdot)$, which is generally larger, i.e.,
  \begin{align} \label{eq_fernando}
    I(A:C|B)_{\rho} \geq \MD\bigl(\rho_{ABC}\dd\cR_{B \to BC}(\rho_{AB})\bigr) \geq -2 \log_2 F\bigl( \rho_{ABC},\cR_{B \to BC}(\rho_{AB})\bigr) \ .
  \end{align}
Furthermore, in~\cite{TB15} an alternative proof of~\eqref{eq_FR} has been derived that uses properties of the fidelity of recovery (in particular, multiplicativity). Another recent work~\cite{BLW14} showed how to generalize ideas from~\cite{FR14} to prove a remainder term for the monotonicity of the relative entropy in terms of a recovery map that satisfies~\eqref{eq_FR}.

All known proofs of~\eqref{eq_FR} are non-constructive, in the sense that the recovery map $\cR_{B \to BC}$ is not given explicitly. 
It is merely known~\cite{FR14} that if $A$, $B$, and $C$ are finite-dimensional then $\cR_{B \to BC}$ can always be chosen such that it has the form
\begin{align} \label{eq_mappingFR}
  X_B \mapsto V_{B C} \rho_{B C}^{\frac{1}{2}} (\rho_B^{-\frac{1}{2}} U_B  X_B U_B^{\dagger} \rho_B^{-\frac{1}{2}} \otimes \id_C) \rho_{B C}^{\frac{1}{2}} V_{B C}^{\dagger} 
\end{align}
on the support of $\rho_B$, where $U_B$ and $V_{BC}$ are unitaries on $B$ and $B\otimes C$, respectively. It would be natural to expect that the choice of the recovery map that satisfies~\eqref{eq_FR} only depends on $\rho_{BC}$, however this is only known in special cases. One such special case are Markov chains $\rho_{ABC}$, i.e., states for which~\eqref{eq_markov} holds perfectly. Here a map of the form~\eqref{eq_mappingFR} with $V_{BC}=\id_{BC}$ and $U_B = \id_B$ (sometimes referred to as \emph{transpose map} or \emph{Petz recovery map}) serves as a perfect recovery map~\cite{Pet86,Pet03}.
Another case where a recovery map that only depends on $\rho_{BC}$ is known explicitly are states with a classical $B$ system, i.e., qcq-states of the form $\rho_{ABC}=\sum_b P_{B}(b) \proj{b}\otimes \rho_{AC,b}$, where $P_B$ is a probability distribution, $\{\ket{b} \}_b$ an orthonormal basis on $B$ and $\{\rho_{AC,b}\}_b$ a family of states on $A \otimes C$. As discussed in~\cite{FR14}, for such states~\eqref{eq_FR} holds for the recovery map defined by $\cR_{B \to BC}(\proj{b})=\proj{b} \otimes \rho_{C,b}$ for all $b$, where $\rho_{C,b}=\tr_A(\rho_{AC,b})$. For general states, however, the previous results left open the possibility that the recovery map $\cR_{B \to BC}$ depends on the \emph{full state} $\rho_{ABC}$ rather than the marginal $\rho_{BC}$ only. In particular, the unitaries $U_B$ and $V_{BC}$ in~\eqref{eq_mappingFR}, although acting only on $B$ respectively $B \otimes C$, could have such a dependence. 

In this work we show that for any state $\rho_{BC}$ on $B \otimes C$ there exists a recovery map $\cR_{B \to BC}$ that is \emph{universal}|in the sense that the distance between \emph{any} extension $\rho_{ABC}$ of $\rho_{BC}$ and $\cR_{B \to BC}(\rho_{AB})$ is bounded from above by the conditional mutual information $I(A:C|B)_{\rho}$. In other words we show that~\eqref{eq_FR} remains valid if the recovery map is chosen depending on $\rho_{BC}$ only, rather than on $\rho_{ABC}$. This result implies a close connection between two different approaches to define topological order of quantum systems.

\section{Main result} \label{sec_mainRes}

\begin{theorem} \label{thm_main}
 For any density operator $\rho_{B C}$ on $B \otimes C$ there exists a trace-preserving completely positive map $\cR_{B \to B C}$ such that for any extension $\rho_{A B C}$ on $A \otimes B \otimes C$
   \begin{align} \label{eq_main_new_separable}
I(A:C|B)_{\rho} \geq - 2 \log_2 F \bigl(\rho_{ABC},\cR_{B \to BC}(\rho_{AB})\bigr) \ ,
  \end{align}
where $A$, $B$, and $C$ are separable Hilbert spaces.
\end{theorem}

\begin{remark} \label{rmk_measRel}
If $B$ and $C$ are finite-dimensional Hilbert spaces, the statement of Theorem~\ref{thm_main} can be tightened to
  \begin{align} \label{eq_main_new}
    I(A:C|B)_{\rho} \geq \MD\bigl(\rho_{ABC}\dd\cR_{B \to BC}(\rho_{AB})\bigr) \ .
  \end{align}
\end{remark}

\begin{remark} \label{rmk_BtoBC}
  The recovery map $\cR_{B \to B C}$ predicted by Theorem~\ref{thm_main} has the property that it maps $\rho_B$ to $\rho_{B C}$. To see this, note that $I(A:C|B)_{\tilde \rho}=0$ for any density operator of the form $\tilde{\rho}_{A B C} = \rho_A\otimes \rho_{B C}$. Theorem~\ref{thm_main} thus asserts that $\tilde{\rho}_{A B C}$ must be equal to~$\cR_{B \to B C}(\tilde{\rho}_{A B})$, which implies that $\rho_{B C} = \cR_{B \to B C}(\rho_B)$. We note that so far it was unknown whether recovery maps that satisfy~\eqref{eq_FR} and have this property do exist.
\end{remark}

We note that Theorem~\ref{thm_main} does not reveal any information about the structure of the recovery map that satisfies~\eqref{eq_main_new_separable}. However, if we consider a linearized version of the bound~\eqref{eq_main_new_separable}, we can make more specific statements.

\begin{corollary} \label{cor_main}
 For any density operator $\rho_{B C}$ on $B\otimes C$ there exists a trace-preserving completely positive map $\cR_{B \to B C}$ such that for any extension $\rho_{A B C}$ on $A \otimes B \otimes C$
  \begin{align} \label{eq_main}
    F\bigl(\rho_{A B C}, \cR_{B \to BC}(\rho_{A B})\bigr) \geq 1- \frac{\ln(2)}{2} I(A : C | B)_{\rho} \ ,
  \end{align}
  where $A$, $B$, and $C$ are separable Hilbert spaces. Furthermore, if $B$ and $C$ are finite-dimensional then $\cR_{B \to BC}$ has the form
  \begin{align} \label{eq_unitalRecMap}
X_B \mapsto \rho_{BC}^{\frac{1}{2}} \, \cU_{BC \to BC}(\rho_B^{-\frac{1}{2}} X_B\rho_B^{-\frac{1}{2}}  \otimes \id_C) \, \rho_{BC}^{\frac{1}{2}}
\end{align}
on the support of $\rho_B$, where $\cU_{BC \to BC}$ is a unital trace-preserving map from $B \otimes C$ to $B \otimes C$.
\end{corollary}

\begin{remark} \label{rmk_diffRepRecMap}
Following the proof of Corollary~\ref{cor_main} we can deduce a more specific structure of the universal recovery map. In the finite-dimensional case the map $\cR_{B \to BC}$ satisfying~\eqref{eq_main} can be assumed to have the form
\begin{align} \label{eq_RecMapFormSpecific}
X_B \mapsto \int V^s_{B C} \rho_{B C}^{\frac{1}{2}} (\rho_B^{-\frac{1}{2}} U^s_B  X_B U^{s\, \dagger}_B \rho_B^{-\frac{1}{2}} \otimes \id_C) \rho_{B C}^{\frac{1}{2}} {V^{s \, \dagger}_{B C}} \, \mu(\mathrm{d}s) \ ,
\end{align}
where $\mu$ is a probability measure on some set $\cS$, $\{V^s_{BC} \}_{s\in \cS}$ is a family of unitaries on $B \otimes C$ that commute with $\rho_{BC}$, and $\{U^s_{B} \}_{s \in \cS}$ is a family of unitaries on $B$ that commute with $\rho_{B}$. However, the representation of the recovery map given in~\eqref{eq_unitalRecMap} has certain advantages compared to the representation~\eqref{eq_RecMapFormSpecific}. The fidelity maximized over all recovery maps of the form~\eqref{eq_unitalRecMap} can be phrased as a semidefinite program and therefore be computed efficiently, whereas it is unknown whether the same is possible for~\eqref{eq_RecMapFormSpecific}.

We note that for almost all density operators $\rho_{BC}$, i.e., for all $\rho_{BC}$ except for a set of measure zero, we can replace the unitaries $U_B^s$ and $V_{BC}^s$ by complex matrix exponentials of the form $\rho_B^{\ci t}$ and $\rho_{BC}^{\ci t}$, respectively, with $t \in \R$. This shows that~\eqref{eq_RecMapFormSpecific} without the integral (the integration in~\eqref{eq_RecMapFormSpecific} is only necessary to guarantee that the recovery map is universal) coincides with the recovery map found in~\cite{wilde15}.\footnote{This follows by the equidistribution theorem which is a special case of the strong ergodic theorem~\cite[Section~II.5]{simon_book} (see also~\cite{einsiedler_book}).}
\end{remark}

\begin{example} \label{cor_MapSpecified}
For density operators with a marginal on $B \otimes C$ of the form $\rho_{BC} = \rho_{B} \otimes \rho_{C}$, a universal recovery map that satisfies~\eqref{eq_main_new} is uniquely defined on the support of $\rho_{B}$|it is the transpose map, which in this case simplifies to $\cR_{B \to BC}  :   X_B \mapsto X_B \otimes \rho_C$. 
It is straightforward to see that \eqref{eq_main_new} holds. In fact, we even have equality if we consider the relative entropy (which is in general larger than the measured relative entropy), i.e.,
\begin{align}
I(A:C|B)_{\rho} = D\bigl(\rho_{ABC}\dd\cR_{B \to BC}(\rho_{AB})\bigr) \ .
\end{align}
The uniqueness of $\cR_{B \to BC}$ on the support of $\rho_B$ follows by using the fact that the universal recovery map should perfectly recover the Markov state $\rho_{AB} \otimes \rho_{C}$ where $\rho_{AB}$ is a purification of $\rho_B$. This forces $\cR_{B \to BC}$ to agree with the transpose map on the support of $\rho_B$~\cite{Pet86,Pet03}.
\end{example}



The proof of Theorem~\ref{thm_main}  is structured into two parts. We first prove the statement for finite-dimensional Hilbert spaces $B$, and $C$ in Section~\ref{sec_finiteDim} and then show that this implies the statement for general separable Hilbert spaces in Section~\ref{sec_infiniteSystems}. The proof of 
Corollary~\ref{cor_main} is given in Section~\ref{sec_pfCor}.

\section{Applications} \label{sec_applications}
A celebrated result known as \emph{strong subadditivity} states that the conditional quantum mutual information of any density operator is non-negative~\cite{LieRus73}, i.e.,
\begin{align}
I(A:C|B)_{\rho} \geq 0 \ ,
\end{align}
for any density operator $\rho_{ABC}$ on $A \otimes B \otimes C$. Theorem~\ref{thm_main} implies a strengthened version of this inequality with a remainder term that is universal in the sense that it only depends on $\rho_{BC}$. The conditional quantum mutual information is a useful tool in different areas of physics and computer science. It is helpful to characterize measures of entanglement~\cite{FR14,LiWin14}, analyze the correlations of quantum many-body systems~\cite{kim_phd,Kim13}, prove quantum de Finetti results~\cite{BraHar13_1,BraHar13_2} and make statements about quantum information complexity~\cite{JRS03,KLLRX12,Touch14}. It is expected that oftentimes when~\eqref{eq_FR} can be used, its universal version (predicted by Theorem~\ref{thm_main}) is even more helpful.   

In the following we sketch an application where the universality result is indispensable.
Theorem~\ref{thm_main} can be applied to establish a connection between two alternative definitions of \emph{topological order of quantum systems} (denoted by TQO and TQO'). Consider an $n$-spin system with $n \in \N$. 
While the following statements should be understood asymptotically, we omit the dependence on $n$ in our notation for simplicity.

According to~\cite{BHV06}, a family of states $\{\rho^{i}\}_{i \in \cI}$ with $\rho^{i} \in \cE$ for all $i \in \cI$ and $|\cI | < \infty$, exhibits topological quantum order (TQO) if and only if any two members of the family: (i) are (asymptotically) orthogonal, i.e., $ F(\rho^{i},\rho^{j}) = 0$ for all $i, j \in \cI$ and (ii) have (asymptotically) the same marginals on any sufficiently small subregion, i.e., $\tr_{G}\rho^i = \tr_{G}\rho^j$ for all $i,j \in \cI$ and $G$ sufficiently large.\footnote{More precisely, we require that $\|\tr_{G}(\rho^i) - \tr_{G}(\rho^j)\|_1 = o(n^{-2})$.}
Alternatively, for three regions $A$, $B$, and $C$ that form a certain topology $\cF$
(see Figure~\ref{fig_topo} and~\cite{KitPres06}), a state $\rho_{A B C}$ on such a subspace exhibits topological quantum order (TQO') if ${I(A:C|B)_{\rho}}=2 \gamma >  0$, where $\gamma$ denotes a \emph{topological entanglement entropy}~\cite{KitPres06,LW06}. 

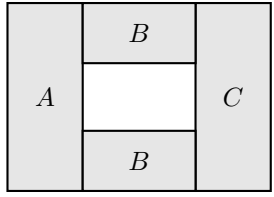
\begin{figure}[!htb]
\centering
\def \xs{1}
\def \ys{0.8}

\def \x{1.5}
\def \y{2.5}

\begin{tikzpicture}[scale=1,auto, node distance=1cm,>=latex']
	
\draw [thick,draw=black, fill=gray!20!, opacity=1]  (0,0) -- (\xs,0) -- (\xs,\y) -- (0,\y) -- cycle;    
 \node at (0.5*\xs,0.5*\y) {$A$};       

\draw [thick,draw=black, fill=gray!20!, opacity=1]  (\xs+\x,0) -- (2*\xs+\x,0) -- (2*\xs+\x,\y) -- (\xs+\x,\y) -- cycle;    
 \node at (1.5*\xs+\x,0.5*\y) {$C$};

\draw [thick,draw=black, fill=gray!20!, opacity=1]  (\xs,0) -- (\xs+\x,0) -- (\xs+\x,\ys) -- (\xs,\ys) -- cycle;    
 \node at (\xs+0.5*\x,0.5*\ys) {$B$};

\draw [thick,draw=black, fill=gray!20!, opacity=1]  (\xs,\y-\ys) -- (\xs+\x,\y-\ys) -- (\xs+\x,\ys+\y-\ys) -- (\xs,\ys+\y-\ys) -- cycle;    
 \node at (\xs+0.5*\x,0.5*\ys+\y-\ys) {$B$};
   
\end{tikzpicture}
\caption{Relevant topology of the subsystems $A$, $B$, and $C$ such that a state $\rho_{ABC}$ exhibits TQO' if ${I(A:C|B)_{\rho}}=2 \gamma >  0$.}
\label{fig_topo}
\end{figure}

It is an open problem to find out how these two characterizations are related, e.g., if a family of states on $\cF$ that exhibits TQO implies that most of its members have TQO'. This connection follows by Theorem~\ref{thm_main}. Suppose $\{\rho^{i} \}_{i \in \cI}$ with $\rho^{i} \in \cF$ for all $i \in \cI$ shows TQO. Then consider subsystems $A$, $B$ and $C$ that together form a non-contractible loop. By definition of TQO, the density operators $\{\rho^{i}\}_{i \in \cI}$ share (asymptotically) the same marginals on ${B \otimes C}$. Applying Theorem~\ref{thm_main} to this common marginal, together with the continuity of the conditional mutual information~\cite{AF03} ensures that there exits a recovery map $\cR_{B \to B C}$ such that for any $i \in \cI$,
\begin{align} \label{eq_topoOrderRec}
 I(A:C|B)_{\rho^{i}} \geq - 2 \log F\bigl( \rho^{i}_{A B C},\cR_{B \to BC}(\rho^{i}_{AB}) \bigr)  \ .
\end{align} 
Since the density operators $\{\rho^{i}_{A B C}\}_{i \in \cI}$ are (asymptotically) orthogonal, share (asymptotically) the same marginals on ${A \otimes B}$, and the fidelity is continuous in its inputs (see Lemma~B.9 in~\cite{FR14}), this implies that for all $i \in \cI$, except of a single element, we have
\begin{align}
 I(A:C|B)_{\rho^i}\geq \mathrm{const} > 0 \ .
\end{align}

\section{Proof for finite dimensions} \label{sec_finiteDim}
Throughout this section we assume that the Hilbert spaces $B$ and $C$ are finite-dimensional. In the proof Steps~\hyperlink{step_3_1}{1} - \hyperlink{step_3_3}{3} below, we also make the same assumption for $A$, but then drop it in Step~\hyperlink{step_3_4}{4}.
We start by explaining why~\eqref{eq_main_new} is a tightened version of~\eqref{eq_main_new_separable} which was noticed in~\cite{BHOS14}.
Let $D_{\alpha}(\cdot\dd\cdot)$ be the \emph{$\alpha$-Quantum R\'enyi Divergence} as defined in~\cite{MLDSFT13,WWY14} with $D_1(\rho \dd \sigma)=D(\rho \dd \sigma):= \tr (\rho(\log \rho - \log \sigma))$.
By definition of the measured relative entropy (see Definition~\ref{def_MeasRelEnt}) we find for any two states $\rho$ and $\sigma$
\begin{multline} \label{eq_fernandoRed}
\MD(\rho\dd\sigma) = \sup_{\cM \in \MM} D\bigl(\cM(\rho)\dd\cM(\sigma) \bigr) \geq \sup_{\cM \in \MM} D_{\frac{1}{2}}\bigl(\cM(\rho)\dd\cM(\sigma) \bigr) = - 2 \log_2 \inf_{\cM \in \MM} F\bigl(\cM(\rho),\cM(\sigma) \bigr) \\
=- 2 \log_2 F(\rho,\sigma) \ ,
\end{multline}
where $\MM:=\{\cM :  \cM(\rho) = \sum_{x} \tr(\rho M_x) \proj{x} \text{ with } \sum_{x} M_x = \id \}$ and $\{\ket{x}\}$ is a family of orthonormal vectors.
The inequality step uses that $\alpha \mapsto D_{\alpha}(\rho\dd\sigma)$ is a monotonically non-decreasing function in $\alpha$~\cite[Theorem~7]{MLDSFT13} and the final step follows from the fact that for any two states there exists an optimal measurement that does not increase their fidelity~\cite[Section 3.3]{Fuc96}.  
As a result, in order to prove Theorem~\ref{thm_main} for finite-dimensional $B$ and $C$ it suffices to prove~\eqref{eq_main_new}.

We first derive a proposition (Proposition~\ref{prop_finite} below) and then show how it can be used to prove~\eqref{eq_main_new} (and, hence, Theorem~\ref{thm_main}).
The proposition refers to a family of functions 
\begin{align} \label{eq_deltaFirst}
\D(A\otimes B \otimes C) \ni \rho \mapsto \Delta_{\cR}(\rho) \in \R \cup \{ - \infty \} \ , 
\end{align}
parameterized by recovery maps $\cR \in \TPCP(B,B \otimes C)$, where $\TPCP(B,B\otimes C)$ denotes the set of trace-preserving completely positive maps from $B$ to $B \otimes C$ and $\D(A\otimes B \otimes C)$ denotes the set of density operators on $A \otimes B \otimes C$. 
Subsequently in the proof, the function family $\Delta_{\cR}(\cdot)$ will be constructed as the difference of the two terms in~\eqref{eq_main_new} (see Equation~\eqref{eq_deltaDefMeasRel}) such that $\Delta_{\cR}(\rho)\geq 0$ corresponds to~\eqref{eq_main_new}. The proposition asserts that if for any extension $\rho_{ABC}$ of $\rho_{BC}$ we have $\Delta_{\cR}(\rho)\geq 0$ for some $\cR \in \TPCP(B,B \otimes C)$ and provided the function family $\Delta_{\cR}(\cdot)$ satisfies certain properties described below, then there exists a single recovery map $\cR$ for which $\Delta_{\cR}(\rho) \geq 0$ for all extensions $\rho_{ABC}$ of $\rho_{BC}$ on a fixed $A$ system. We note that the precise form of the function family $\Delta_{\cR}(\cdot)$ is irrelevant for Proposition~\ref{prop_finite} as long as it satisfies a list of properties as stated below.

As described above, our goal is to prove that there exists a recovery map $\cR_{B \to BC}$ such that ${\Delta_{\cR}(\rho) \geq 0}$ for all $\rho_{ABC} \in \D(A\otimes B \otimes C)$ with a fixed marginal $\rho_{BC}$ on $B \otimes C$. To formulate our argument more concisely, we introduce some notation. For any set $\cS$ of density operators $\rho_{ABC} \in  \D(A\otimes B \otimes C)$ we define
\begin{align}
\Delta_{\cR}(\cS):=\inf_{\rho \in \cS}\Delta_{\cR}(\rho) \ .
\end{align}
The desired statement then reads as $\Delta_{\cR}(\cS) \geq 0$, for any set $\cS$ of states on $A \otimes B \otimes C$ with a fixed marginal $\rho_{BC}$. Furthermore, for any fixed states $\rho^0_{ABC}$ and $\rho_{ABC}$ on $A \otimes B \otimes C$ and $p \in [0,1]$, we define
    \begin{align} \label{eq_rhop}
        \rho^p_{\hat{A} A B C} := (1-p) \proj{0}_{\hat{A}} \otimes \rho^0_{A B C} + p \proj{1}_{\hat{A}} \otimes \rho_{A B C} \ , 
    \end{align}    
where $\hat{A}$ is an additional system with two orthogonal states $\ket{0}$ and $\ket{1}$.    
More generally, for any fixed state $\rho^0_{ABC}$ and for any set $\cS$ of density operators $\rho_{ABC}$ we set 
\begin{align} \label{eq_Sp}
\cS^p := \bigl\{\rho^p_{\hat{A} A B C} : \, \rho_{A B C} \in \cS \bigr\} \ .
\end{align}
\begin{property} \ 
\begin{enumerate}
\item For any $\rho^0_{ABC}, \rho_{ABC} \in \D(A\otimes B \otimes C)$ with identical marginals $\rho^0_{BC} = \rho_{BC}$ on $B \otimes C$, for any $\cR \in \TPCP(B,B\otimes C)$, and for any $p \in [0,1]$ we have $\Delta_{\cR}(\rho^p) = (1-p) \Delta_{\cR}(\rho^0) + p \Delta_{\cR}(\rho)$. \label{property_1}
\item For any $\cR, \cR' \in \TPCP(B,B\otimes C)$, for any $\alpha \in [0,1]$, and $\bar \cR=\alpha \cR + (1-\alpha)\cR'$ we have $\Delta_{\bar \cR}(\rho) \geq \alpha \Delta_{\cR}(\rho) + (1-\alpha) \Delta_{\cR'}(\rho)$ for all $\rho \in \D(A \otimes B \otimes C)$.  \label{property_2}

\item For any $\cR \in \TPCP(B,B\otimes C)$, the function $\D(A \otimes B \otimes C) \ni \rho \mapsto \Delta_{\cR}(\rho) \in \R\cup \{-\infty\}$ is upper semicontinuous.\label{property_4}

\item For any $\rho \in \D(A\otimes B \otimes C)$, the function $\TPCP(B,B\otimes C) \ni \cR \mapsto \Delta_{\cR}(\rho) \in \R\cup \{-\infty\}$ is upper semicontinuous.\label{property_5}

\end{enumerate} 
\end{property}

Property~\ref{property_1} implies that for any state $\rho^0_{ABC}$, for any set $\cS$ of operators $\rho_{ABC}$ with $\rho_{BC} = \rho^0_{BC}$, and for any $p \in [0,1]$ we have
\begin{align} \label{eq_implicationProp1}
\Delta_{\cR}(\cS^p) =\inf \limits_{\rho \in \cS} \Delta_{\cR}(\rho^p) 
= (1-p) \Delta_{\cR}(\rho^0) + p \inf \limits_{\rho \in \cS} \Delta_{\cR}(\rho)
= (1-p) \Delta_{\cR}(\rho^0) + p  \Delta_{\cR}(\cS) \ .
\end{align}
Similarly, Property~\ref{property_2} implies
\begin{multline} \label{eq_implicationProp2}
  \Delta_{\bar{\cR}}(\cS) 
  = \inf_{\rho \in \cS} \Delta_{\bar{\cR}}(\rho) 
  \geq \inf_{\rho \in \cS}  \bigl\{ \alpha \Delta_{\cR}(\rho) + (1-\alpha) \Delta_{\cR'}(\rho) \bigr\} \\
  \geq \alpha \inf_{\rho \in \cS}  \Delta_{\cR}(\rho) + (1-\alpha) \inf_{\rho \in \cS}  \Delta_{\cR'}(\rho)
  = \alpha \Delta_{\cR}(\cS) +  (1-\alpha) \Delta_{\cR'}(\cS) \ .
\end{multline}

\begin{proposition} \label{prop_finite}
Let $A$, $B$, and $C$ be finite-dimensional Hilbert spaces, $\cP \subseteq \TPCP(B,B\otimes C)$ be compact and convex, $\cS$ be a set of density operators on $A \otimes B \otimes C$ with identical marginals on $B \otimes C$, and $\Delta_{\cR}(\cdot)$ be a family of functions of the form~\eqref{eq_deltaFirst} that satisfies Properties~\ref{property_1}-\ref{property_5}. Then
\begin{align} \label{eq_prop}
\forall \rho\in \cS \,\,  \exists \cR \in \cP \ : \ \Delta_{\cR}(\rho)\geq0 \quad \implies \quad \exists \bar{\cR} \in \cP \ : \ \Delta_{\bar{\cR}}(\cS) \geq 0 \ .
\end{align}
\end{proposition}

We now proceed in four steps. In the first, we prove Proposition~\ref{prop_finite} for finite sets $\cS$. This is done by induction over the cardinality of the set $\cS$. We show that if the statement of Proposition~\ref{prop_finite} is true for all sets $\cS$ with $|\cS|=n$, this implies that it remains valid for all sets $\cS$ with $|\cS|=n+1$.
In Step~\hyperlink{step_3_2}{2}, we use an approximation step to extend this to infinite sets $\cS$ which then completes the proof of Proposition~\ref{prop_finite}. In the final two steps, we show how to conclude the statement of Theorem~\ref{thm_main} for the finite-dimensional case from that. In Step~\hyperlink{step_3_3}{3} we prove~\eqref{eq_main_new} for the case where the recovery map that satisfies~\eqref{eq_main_new} could still depend on the dimension of the system $A$. In Step~\hyperlink{step_3_4}{4} we show how this dependency can be removed.

\subsection*{\hypertarget{step_3_1}{Step 1: }Proof of Proposition~\ref{prop_finite} for finite size sets $\cS$}
We proceed by induction over the cardinality $n:= |\cS|$ of the set $\cS$ of density operators. More precisely, the induction hypothesis is that for any finite-dimensional Hilbert space $A$ and any set $\cS$ of size $n$ consisting of density operators on $A \otimes B \otimes C$ with fixed marginal $\rho_{B C}$ on $B \otimes C$, the statement~\eqref{eq_prop} holds. 
For $n=1$, this hypothesis holds trivially for $\bar \cR = \cR$.\footnote{For $n=0$, we have $\Delta_{\cR}(\cS)=\infty\geq 0$ for any $\cR \in \cP$ since the infimum of an empty set is infinity.} 

We now prove the induction step. Suppose that the induction hypothesis holds for some $n$. Let $A$ be a finite-dimensional Hilbert space and let $\cS \cup \{\rho_{ABC}^0\}$ be a set of cardinality $n+1$ where $\cS$ is a set of states on $A\otimes B \otimes C$ with fixed marginal $\rho_{BC}$ on $B \otimes C$ of cardinality $n$ and $\rho^0_{ABC}$ is another state with $\rho^0_{BC}=\rho_{BC}$. 
We need to prove that there exists a recovery map $\bar \cR_{B \to BC} \in \cP$ such that 
\begin{align} \label{eq_inductionstep}
  \Delta_{\bar \cR}(\cS \cup \{\rho_{ABC}^0\})  \geq 0 \ .
\end{align}

Let $p \in [0,1]$ and consider the set $\cS^p$ as defined in~\eqref{eq_Sp}.
In the following we view the states $\rho^p$ (see Equation~\eqref{eq_rhop}) in this set as tripartite states on $(\hat A \otimes A) \otimes B \otimes C$, i.e., we regard the system $\hat A \otimes A$ as one (larger) system. 
The induction hypothesis applied to the extension space $\hat A \otimes A$ and the set $\cS^p$ (of size $n$) of states on $(\hat A  \otimes A) \otimes B \otimes C$ implies the existence of a map $\cR^{p}_{B \to BC} \in \cP$ such that
\begin{align}
\Delta_{\cR^{p}}(\cS^p) \geq 0 \ .
\end{align}
As by assumption the function $\D(A\otimes B \otimes C) \ni \rho \mapsto \Delta_{\cR^p}(\rho) \in \R\cup \{-\infty\}$ satisfies Property~\ref{property_1} (and hence also Equation~\eqref{eq_implicationProp1}) we obtain
\begin{align} \label{eq_DeltaCQdav}
(1-p) \Delta_{\cR^{p}}(\rho^0) + p \Delta_{\cR^{p}}(\cS) \geq 0 \ .
\end{align}
This implies that
\begin{align}
\Delta_{\cR^{p}}(\rho^0) \geq 0 \quad \textnormal{or} \quad \Delta_{\cR^{p}}(\cS) \geq 0 \ .
\end{align}
Furthermore, for $p=0$ the left inequality holds and for $p=1$ the right inequality holds. By choosing $K_0=\{p \in [0,1]: \Delta_{\cR^p}(\rho^0) \geq 0 \}$ and $K_1=\{ p \in [0,1] : \Delta_{\cR^p}(\cS) \geq 0 \}$, Lemma~\ref{lem_meanVal} implies that for any $\delta>0$ there exist $u,v \in[0,1]$ with $0\leq v-u \leq \delta$ such that
\begin{align} \label{eq_DeltaPos}
\Delta_{\cR^{u}}(\rho^0) \geq 0 \quad \textnormal{and} \quad \Delta_{\cR^{v}}(\cS) \geq 0 \ .
\end{align}
Note also that $\cR_{B \to BC}^{u}$, $\cR_{B \to BC}^{v} \in \cP$, since by the induction hypothesis $\cR^{p}_{B \to BC} \in \cP$ for any $p\in [0,1]$.

We will use this to prove that  the recovery map $\tilde \cR \in \cP$ defined by
\begin{align}
\tilde \cR:=\alpha \cR^{u} + (1-\alpha) \cR^{v} \ ,
\end{align}
for an appropriately chosen $\alpha \in [0,1]$, satisfies
\begin{align} \label{eq_toprove}
\Delta_{\tilde \cR}(\rho^0) \geq - c \delta \quad \textnormal{and} \quad \Delta_{\tilde \cR}(\cS) \geq - c \delta \ ,
\end{align}
where $c$ is a constant defined by 
\begin{align} \label{eq_constC}
c:=4 \max_{\cR \in \TPCP(B,B\otimes C)} \max_{\rho \in \D(A\otimes B \otimes C)}\Delta_{\cR}(\rho) < \infty \ .
\end{align}
Properties~\ref{property_4} and~\ref{property_5} together with Lemma~\ref{lem_DensityOpCompact} and Remark~\ref{rmk_TPCPcompact} ensure that the two maxima in~\eqref{eq_constC} are attained which implies by the definition of the codomain of $\Delta_{\cR}(\cdot)$ (see Equation~\eqref{eq_deltaFirst}) that $c$ is finite.
In other words, for any $\delta > 0$ there exists a recovery map $\tilde \cR^{\delta} \in \cP$ such that
\begin{align} \label{eq_deltaRec}
 \Delta_{\tilde \cR^{\delta}}(\cS \cup \{\rho^0\}) \geq - c \delta \ .
\end{align}
The compactness of $\cP$ ensures that there exists a recovery map $\bar \cR \in \cP$ and a sequence $\{\delta_n\}_{n \in \N}$ such that
\begin{align} \label{eq_seqCom}
\lim_{n \to \infty} \delta_n = 0 \quad \textnormal{and} \quad \lim_{n \to \infty} \tilde \cR^{\delta_n} = \bar \cR \ .
\end{align}
Because of~\eqref{eq_deltaRec} we have
\begin{align}
\limsup_{n \to \infty} \Delta_{\tilde \cR^{\delta_n}}(\cS \cup \{\rho^0\}) \geq \lim_{n \to \infty} - c \delta_n =0  \ , 
\end{align}
which together with Property~\ref{property_5} implies that
\begin{multline}
\Delta_{\bar \cR}(\cS \cup \{\rho^0\})
=\min_{\rho \in \cS \cup \{\rho^0\}}\Delta_{\bar \cR}(\rho) 
\geq \min_{\rho \in \cS \cup \{\rho^0\}} \limsup_{n \to \infty} \Delta_{\tilde \cR^{\delta_n}}(\rho) 
\geq \limsup_{n \to \infty} \min_{\rho \in \cS \cup \{\rho^0\}} \Delta_{\tilde \cR^{\delta_n}}(\rho) \\
= \limsup_{n \to \infty} \Delta_{\tilde \cR^{\delta_n}}(\cS \cup \{\rho^0\}) 
\geq 0 \ ,
\end{multline}
and thus proves~\eqref{eq_inductionstep}.

It thus remains to show~\eqref{eq_toprove}. To simplify the notation let us define
\begin{align}
\Lambda^0:=\Delta_{\cR^u}(\rho^0) \quad \textnormal{and} \quad \Lambda^1:=\Delta_{\cR^v}(\cS)
\end{align}
as well as
\begin{align}
\bar \Lambda^0:=\Delta_{\cR^v}(\rho^0) \quad \textnormal{and} \quad \bar \Lambda^1:=\Delta_{\cR^u}(\cS) \ .
\end{align}
It follows from~\eqref{eq_DeltaCQdav} that 
\begin{align} \label{eq_one_dav}
(1-u) \Lambda^0 + u \bar \Lambda^1  \geq 0 \ .
\end{align}
Similarly, we have
\begin{align}\label{eq_two_dav}
(1-v)\bar \Lambda^0 + v \Lambda^1  \geq 0 \ .
\end{align}
As by assumption the function $\Delta_{\cR}(\cdot)$ satisfies Property~\ref{property_2} we find together with~\eqref{eq_two_dav} that for any $\alpha \in [0,1]$ and $\bar \cR = \alpha \cR^u + (1-\alpha) \cR^v$,
\begin{align} \label{eq_done1}
\Delta_{\bar \cR}(\rho^0) \geq \alpha \Delta_{\cR^u}(\rho^0) + (1-\alpha) \Delta_{\cR^v}(\rho^0) = \alpha \Lambda^0 + (1-\alpha) \bar \Lambda^0 \geq \alpha \Lambda^0 - (1-\alpha) \frac{v}{1-v} \Lambda^1 \ .
\end{align}
(If $v = 1$ it suffices to consider the case $\alpha = 1$ so that the last term can be omitted; cf.~Equation~\eqref{eq_alphatwo} below.)
Analogously, using~\eqref{eq_implicationProp2} and~\eqref{eq_one_dav}, we find
\begin{align} \label{eq_done2}
\Delta_{\bar \cR}(\cS) \geq \alpha \Delta_{\cR^u} (\cS) + (1-\alpha) \Delta_{\cR^v}(\cS) = \alpha \bar \Lambda^1 + (1-\alpha) \Lambda^1 \geq - \alpha \frac{1-u}{u} \Lambda^0 + (1-\alpha) \Lambda^1 \ .
\end{align}
(If $u=0$ it suffices to consider the case $\alpha = 0$; cf.~Equation~\eqref{eq_alphaone} below.)

To conclude the proof of~\eqref{eq_toprove}, it suffices to choose $\alpha \in [0,1]$  such that the terms on the right hand side of~\eqref{eq_done1} and~\eqref{eq_done2} satisfy
\begin{align}
  \alpha \Lambda^0 - (1-\alpha) \frac{v}{1-v} \Lambda^1 & \geq - c \delta \label{eq_dzero1} 
  \end{align}
  and
  \begin{align}
  - \alpha \frac{1-u}{u} \Lambda^0 + (1-\alpha) \Lambda^1 & \geq -c \delta  \label{eq_dzero2} \ .
\end{align}
Let us first assume that $u \geq \frac{1}{2}$. Since $\Lambda^0$ and $\Lambda^1$ are non-negative (see Equation~\eqref{eq_DeltaPos}), we may choose  $\alpha \in [0,1]$ such that
\begin{align} \label{eq_alphatwo}
\alpha (1-v) \Lambda^0 = (1-\alpha)v \Lambda^1 \ .
\end{align}
This immediately implies that the left hand side of~\eqref{eq_dzero1} equals~$0$, so that the inequality holds. As $\frac{1}{2} \leq u\leq v \leq 1$ and $v-u \leq \delta$ we have
\begin{align} \label{eq_implicDeltaClose}
\left| \frac{1-u}{u}-\frac{1-v}{v} \right| \leq 4 \delta \ .
\end{align}
Combining this with~\eqref{eq_alphatwo} we find
\begin{align}
 - \alpha \frac{1-u}{u} \Lambda^0 + (1-\alpha) \Lambda^1 & \geq - \alpha \Lambda^0 \Bigl( \frac{1-v}{v}+ 4 \delta \Bigr)  + (1-\alpha) \Lambda^1 =- 4 \alpha \Lambda^0 \delta \geq - 4 \Lambda^0 \delta\ ,
\end{align}
which proves~\eqref{eq_dzero2} because by~\eqref{eq_constC} we have $\Lambda^0 \leq \frac{c}{4}$. 

Analogously,  if $u < \frac{1}{2}$, choose $\alpha \in [0,1]$ such that
\begin{align} \label{eq_alphaone}
\alpha (1-u) \Lambda^0 = (1-\alpha)u \Lambda^1 \ .
\end{align}
This immediately implies that the left hand side of~\eqref{eq_dzero2} equals~$0$, so that the inequality holds. Furthermore, for $\delta > 0$ sufficiently small such that $v \leq \frac{1}{2}$, we obtain
\begin{align} \label{eq_imp2Delta}
\left|  \frac{v}{1-v} - \frac{u}{1-u} \right| < 4\delta \ .
\end{align}
Together with~\eqref{eq_alphaone} this implies
\begin{align}
 \alpha \Lambda^0 - (1-\alpha) \frac{v}{1-v} \Lambda^1 \geq \alpha \Lambda^0 - (1-\alpha) \Lambda^1 \Bigl(\frac{u}{1-u}+4\delta \Bigr) = - 4 (1-\alpha) \Lambda^1 \delta \geq -4 \Lambda^1 \delta \ ,
\end{align}
which establishes~\eqref{eq_dzero1}. This concludes the proof of Proposition~\ref{prop_finite} for sets $\cS$ of finite size.

\subsection*{\hypertarget{step_3_2}{Step 2: }Extension to infinite sets $\cS$}
All that remains to be done to prove Proposition~\ref{prop_finite} is to generalize the statement to arbitrarily large sets $\cS$. In fact, we show that there exists a recovery map $\cR_{B \to BC} \in \cP$ such that $\Delta_{\cR}(\cS) \geq 0$, where $\cS$ is the set of all density operators on $A \otimes B \otimes C$ for a fixed finite-dimensional Hilbert space $A$ and a fixed marginal $\rho_{BC}$. 

Note first that this set $\cS$ of all density operators on $A \otimes B \otimes C$ with fixed marginal $\rho_{B C}$ on $B \otimes C$ is compact (see Lemma~\ref{lem_DensityOpCompactFixedMarginal}).  
This implies that for any $\eps>0$ there exists a finite set $\cS^\varepsilon$ of density operators on $A \otimes B \otimes C$ such that any $\rho \in \cS$ is $\varepsilon$-close to an element of $\cS^\varepsilon$. 
We further assume without loss of generality that $\cS^{\eps'} \subset \cS^{\eps}$ for $\eps' \geq \eps$. Let $\cR^{\eps} \in \TPCP(B,B\otimes C)$ be a map such that $\Delta_{\cR^\eps}(\cS^\eps) \geq 0$, whose existence follows from the validity of Proposition~\ref{prop_finite} for sets of finite size (which we proved in Step~\hyperlink{step_3_1}{1}).
Since the set $\TPCP(B,B\otimes C)$ is compact (see Remark~\ref{rmk_TPCPcompact}) there exists a decreasing sequence $\{\eps_n \}_{n \in \N}$ and $\bar \cR \in \TPCP(B,B\otimes C)$ such that 
\begin{align}
\lim_{n \to \infty} \eps_n =0 \quad \textnormal{and} \quad \bar \cR = \lim_{n \to \infty} \cR^{\eps_n} \ .
\end{align}
Combining this with Property~\ref{property_5} gives for all $n \in \N$
\begin{multline} \label{eq_claimRR}
\Delta_{\bar \cR}(\cS^{\eps_n}) = \inf_{\rho \in \cS^{\eps_n}} \Delta_{\bar \cR}(\rho) 
\geq \inf_{\rho \in \cS^{\eps_n}} \limsup_{m \to \infty} \Delta_{\bar \cR^{\eps_m}}(\rho) 
\geq \limsup_{m \to \infty}   \inf_{\rho \in \cS^{\eps_n}} \Delta_{\bar \cR^{\eps_m}}(\rho)\\
\geq \limsup_{m \to \infty}  \inf_{\rho \in \cS^{\eps_m}} \Delta_{\bar \cR^{\eps_m}}(\rho)
=\limsup_{m \to \infty}  \Delta_{\cR^{\eps_m}}(\cS^{\eps_m}) \geq 0 \ ,
\end{multline}
where the third inequality holds since $\cS^{\eps_n} \subset \cS^{\eps_m}$ for $\eps_n \geq \eps_m$, respectively $n \leq m$. The final inequality follows from the defining property of $\cR^\eps$. For any fixed $\rho \in \cS$ and for all $n \in \N$, let $\rho^n \in \cS^{\eps_n}$ be such that $ \lim_{n \to \infty} \rho^{n} = \rho \in \cS$. (By definition of $\cS^{\eps_n}$ it follows that such a sequence $\{\rho^n\}_{n \in \N}$ with $\rho^n \in \cS^{\eps_n}$ always exists.) Property~\ref{property_4} together with~\eqref{eq_claimRR}  yields
\begin{align} \label{eq_esssDone}
\Delta_{\bar \cR}(\rho) = \Delta_{\bar \cR}(\lim_{n \to \infty} \rho^n) \geq \limsup_{n \to \infty} \Delta_{\bar \cR}(\rho^n) \geq \limsup_{n \to \infty} \Delta_{\bar \cR}(\cS^{\eps_n}) \geq 0 \ .
\end{align}
Since~\eqref{eq_esssDone} holds for any $\rho \in \cS$, we obtain $\Delta_{\bar \cR}(\cS) \geq 0$, which completes the proof of Proposition~\ref{prop_finite}.

\subsection*{\hypertarget{step_3_3}{Step 3: }From Proposition~\ref{prop_finite} to Theorem~\ref{thm_main} for fixed system $A$}
We next show that Theorem~\ref{thm_main}, for the case where $A$ is a fixed finite-dimensional system, follows from Proposition~\ref{prop_finite}. For this we use Proposition~\ref{prop_finite} for the function family
\begin{align} \label{eq_deltaDefMeasRel}
\Delta_{\cR}: \, \,\, \, &\D(A \otimes B \otimes C) \to \R\cup \{-\infty\} \nonumber \\
& \rho_{ABC} \mapsto I(A:C|B)_{\rho} - \MD\bigl(\rho_{ABC},\cR_{B \to BC}(\rho_{AB})\bigr)\ ,
\end{align}
with $\cR_{B \to BC} \in \TPCP(B,B\otimes C)$. We note that since $C$ is finite-dimensional this implies that $\Delta_{\cR}(\rho) < \infty$ for all $\rho \in \D(A \otimes B \otimes C)$.
To apply Proposition~\ref{prop_finite}, we have to verify that the function family $D(A\otimes B \otimes C) \ni \rho \mapsto \Delta_{\cR}(\rho) \in \R\cup \{-\infty\}$ of the form~\eqref{eq_deltaDefMeasRel} satisfies the assumptions of the proposition. This is ensured by the following lemma.

\begin{lemma} \label{lem_deltaMeasRelProp1}
Let $A$ be a separable and $B$ and $C$ finite-dimensional Hilbert spaces. The function family $\Delta_{\cR}(\cdot)$ defined by~\eqref{eq_deltaDefMeasRel} satisfies Properties~\ref{property_1}-\ref{property_5}.
\end{lemma}
\begin{proof}
We first verify that the function $\Delta_{\cR}(\cdot)$ satisfies Property~\ref{property_1}. For any state $\rho^p$ of the form~\eqref{eq_rhop}, we have by the definition of the mutual information
\begin{align}
I(\hat A A :C|B)_{\rho^p} = H(C|B)_{\rho^p} - H(C|B A \hat{A})_{\rho^p} \ .
\end{align}
Because $\rho^0_{B C} = \rho_{B C}$, the first term, $H(C|B)_{\rho^p}$, is independent of $p$, i.e., $H(C|B)_{\rho^p} = H(C|B)_{\rho^0} =H(C|B)_{\rho}$. The second term can be written as an expectation over $\hat{A}$, i.e., 
\begin{align}
H(C|B A \hat{A})_{\rho^p} = {(1-p)}H(C|B A)_{\rho^0}  + p H(C|B A)_{\rho} \ .
\end{align}
As a result we find  
    \begin{align} \label{eq_measDec}
      I(\hat A A :C|B)_{\rho^p} = (1-p) I(A:C|B)_{\rho^0} + p I(A:C|B)_{\rho} \ .
    \end{align}
   The density operator $ \cR_{B \to B C}(\rho^p_{\hat{A} A B})$ can be written as
    \begin{align}
   \cR_{B \to B C}(\rho^p_{\hat{A} A B}) = (1-p) \proj{0}_{\hat{A}} \otimes \cR_{B \to B C}(\rho^0_{A B}) + p \proj{1}_{\hat{A}} \otimes \cR_{B \to B C}(\rho_{A B}) \ .
   \end{align}    
We can thus apply Lemma~\ref{lem_cqMeasRelEnt}, from which we obtain   
 \begin{align} \label{eq_measRelEntDec}
 \MD\bigl(\rho_{\hat A ABC}^p \dd \cR_{B \to BC}(\rho_{\hat A A B}^p)\bigr) = (1-p) \MD\bigl(\rho_{ ABC}^0 \dd \cR_{B \to BC}(\rho_{ A B}^0)\bigr) + p  \MD\bigl(\rho_{ ABC}^p \dd \cR_{B \to BC}(\rho_{ A B}^p)\bigr) \ .
 \end{align} 
    Equations~\eqref{eq_measDec} and~\eqref{eq_measRelEntDec} imply that
    \begin{align} \label{eq_DeltadecMeas}
      \Delta_{\cR}(\rho^p) = (1-p) \Delta_{\cR}(\rho^0) + p \Delta_{\cR}(\rho) \ ,
    \end{align}
which concludes the proof of Property~\ref{property_1}.    
    
That $\Delta_{\cR}(\cdot)$ satisfies Property~\ref{property_2} can be seen as follows. Let $\cR_{B \to BC}, \cR'_{B \to BC} \in \TPCP(B,B\otimes C)$, $\alpha \in [0,1]$ and $\bar \cR_{B \to BC} = \alpha \cR_{B \to BC} + (1-\alpha) \cR'_{B \to BC}$. Lemma~\ref{lem_convexityMeasRelEnt} implies that for any state $\rho_{ABC}$ on $A \otimes B \otimes C$ we have
\begin{multline}
\MD\bigl(\rho_{ABC} \dd \bar \cR_{B \to BC}(\rho_{AB}) \bigr) = \MD\bigl(\rho_{ABC} \dd \alpha \cR_{B \to BC}(\rho_{AB}) + (1-\alpha) \cR'_{B \to BC}(\rho_{AB}) \bigr) \\
 \leq \alpha \MD\bigl(\rho_{ABC} \dd \cR_{B \to BC}(\rho_{AB}) \bigr) + (1-\alpha)\MD\bigl(\rho_{ABC} \dd \cR'_{B \to BC}(\rho_{AB}) \bigr)
\end{multline}
and hence
\begin{align}
\Delta_{\bar \cR}(\rho) \geq \alpha \Delta_{\cR}(\rho) + (1-\alpha) \Delta_{\cR'}(\rho) \ .
\end{align}    

We next verify that the function $\Delta_{\cR}(\cdot)$ satisfies Property~\ref{property_4}. The Alicki-Fannes inequality ensures that $\D(A \otimes B \otimes C) \ni \rho \mapsto {I(A:C|B)_{\rho}} \in \R^+$ is continuous since $C$ is finite-dimensional~\cite{AF03}. By the definition of $\Delta_{\cR}(\cdot)$ it thus suffices to show that $\D(A \otimes B \otimes C) \ni \rho_{ABC} \mapsto \MD(\rho_{ABC}\dd \cR_{B \to BC}(\rho_{AB})) \in \R^+$ is lower semicontinuous.
  Let $\{\rho_{ABC}^n\}_{n \in \N}$ be a sequence of states on $A \otimes B \otimes C$ such that $\lim_{n \to \infty} \rho_{ABC}^n = \rho_{ABC} \in {\D(A \otimes B \otimes C)}$. 
 By definition of the measured relative entropy (see Definition~\ref{def_MeasRelEnt}), we find for $\MM:=\{\cM :  \cM(\rho) = \sum_{x} \tr(\rho M_x) \proj{x} \text{ with } \sum_{x} M_x = \id \}$,
\begin{align}
\liminf_{n \to \infty} \MD\bigl(\rho_{ABC}^n\dd\cR_{B\to BC}(\rho_{AB}^n)\bigr) 
&= \liminf_{n \to \infty} \sup_{\cM \in \MM} D\Bigl(\cM(\rho_{ABC}^n)\dd\cM\bigl(\cR_{B \to BC}(\rho_{AB}^n)\bigr)\Bigr) \nonumber \\
&\geq \sup_{\cM \in \MM} \liminf_{n \to \infty} D\Bigl(\cM(\rho_{ABC}^n)\dd\cM\bigl(\cR_{B \to BC}(\rho_{AB}^n)\bigr)\Bigr) \nonumber\\
 &\geq \sup_{\cM \in \MM} D\Bigl(\cM(\rho_{ABC})\dd\cM\bigl(\cR_{B \to BC}(\rho_{AB})\bigr)\Bigr)\nonumber \\
 &= \MD\bigl(\rho_{ABC}\dd\cR_{B \to BC}(\rho_{AB})\bigr)  \ .
\end{align}
In the penultimate step, we use that the relative entropy is lower semicontinuous~\cite[Exercise~7.22]{holevo_book} and that $\cM$ as well as $\cR_{B \to BC}$ are linear and bounded operators and hence continuous.

We finally show that $\Delta_{\cR}(\cdot)$ fulfills Property~\ref{property_5}. It suffices to verify that ${\TPCP(B,B\otimes C)} \ni \cR \mapsto \MD(\rho_{ABC} \dd \cR(\rho_{AB})) \in \R^+$ is lower semicontinuous where by definition of the measured relative entropy (see Definition~\ref{def_MeasRelEnt}) we have $\MD(\rho_{ABC} \dd \cR(\rho_{AB})) = \sup_{\cM \in \MM}D(\cM(\rho_{ABC})\dd \cM(\cR_{B \to BC}(\rho_{AB}))) $. Note that since $\cR$ and $\cM$ are linear bounded operators and hence continuous and the relative entropy for two states $\sigma_1$ and $\sigma_2$ is defined by $D(\sigma_1 \dd \sigma_2):=\tr(\sigma_1 (\log \sigma_1 - \log \sigma_2))$ we find that $\cR \mapsto D(\cM(\rho_{ABC})\dd \cM(\cR_{B \to BC}(\rho_{AB})))$ is continuous as the logarithm $\R^+ \ni x \mapsto \log x \in \R$ is continuous. Since the supremum of continuous functions is lower semicontinuous~\cite[Chapter~IV, Section~6.2, Theorem~4]{bourbaki_book}, the assertion follows. 

\end{proof}

What remains to be shown in order to apply Proposition~\ref{prop_finite} is that for any $\rho \in \cS$ where $\cS$ is the set of states on $A \otimes B \otimes C$ with a fixed marginal $\rho_{BC}$ on $B \otimes C$, there exists a recovery map $\cR_{B \to BC} \in \cP$ such that $\Delta_{\cR}(\rho) \geq 0$. By choosing $\cP=\TPCP(B,B\otimes C)$, the main result of~\cite{BHOS14} however precisely proves this. We have thus shown that $\Delta_{\cR}(\rho) \geq 0$ holds for a universal recovery map $\cR_{B \to BC} \in \cP$, so that~\eqref{eq_main_new} follows for any fixed dimension of the $A$ system.
This proves the statement of Remark~\ref{rmk_measRel} (and, hence, Theorem~\ref{thm_main}) for the case where $A$ is a fixed finite-dimensional Hilbert space. 

\subsection*{\hypertarget{step_3_4}{Step 4: }Independence from the $A$ system}

Let $\cS$ be the set of all density operators on $\bar A \otimes B \otimes C$ with a fixed marginal $\rho_{BC}$ on $B \otimes C$, where $B$ and $C$ are finite-dimensional Hilbert spaces and $\bar A$ is the infinite-dimensional Hilbert space $\ell^2$ of square summable sequences. We now show that there exists a recovery map $\cR_{B \to BC}$ such that $\Delta_{\cR}(\cS) \geq 0$.

Let $\{\Pi_{\bar{A}}^{a}\}_{a \in \N}$ be a sequence of finite-rank projectors on $\bar A$ that converges to $\id_{\bar A}$ with respect to the weak operator topology.
Let $\cS^a$ denote the set of states whose marginal on $\bar A$ is contained in the support of $\Pi^a_{\bar A}$ and with the same fixed marginal $\rho_{BC}$ on $B \otimes C$ as the elements of $\cS$.
For all $a \in \N$, let $\cR^{a}_{B\to BC}$ denote a recovery map that satisfies $\Delta_{\cR^{a}}( \cS^{a})\geq 0$. Note that the existence of such maps is already established by the proof of Theorem~\ref{thm_main} for the finite-dimensional case.
As the set of trace-preserving completely positive maps on finite-dimensional systems is compact (see Remark~\ref{rmk_TPCPcompact}) there exists a subsequence $\{ a_i\}_{i \in \N}$ such that $\lim_{i \to \infty} a_i = \infty$ and $\lim_{i \to \infty} \cR^{a_i}= \bar \cR \in \TPCP(B,B\otimes C) $. 
 For every $\rho \in  \cS$ there exists a sequence of states $\{\rho^a\}_{a\in \N}$ with $\rho^a \in \cS^a$ that converges to $\rho$ in the trace norm (see Lemma~\ref{lem_projSequence}).
Lemma~\ref{lem_deltaMeasRelProp1} (in particular Properties~\ref{property_4} and~\ref{property_5}), yields for any $\rho \in  \cS$
\begin{multline}
\Delta_{\bar \cR}(\rho) \geq \limsup_{a \to \infty} \Delta_{\bar \cR}(\rho^a) 
\geq \limsup_{a \to \infty} \limsup_{i \to \infty} \Delta_{\cR^{a_i}}(\rho^a) 
\geq \limsup_{a \to \infty} \limsup_{i \to \infty} \inf_{\rho \in \cS^a} \Delta_{\cR^{a_i}}(\rho)\\
\geq \limsup_{i \to \infty} \inf_{\rho \in \cS^{a_i}} \Delta_{\cR^{a_i}}(\rho)
=  \limsup_{i \to \infty} \Delta_{\cR^{a_i}}( \cS^{a_i}) \geq 0 \ .
\end{multline}
The fourth inequality follows since $a_i \geq a$ for large enough $i$ and since this implies that $ \cS^{a_i} \supset  \cS^a$, and the final inequality follows by definition of $\cR^{a_i}$. This shows that $\Delta_{\bar \cR}(\cS) \geq 0$. 

To retrieve the statement of Remark~\ref{rmk_measRel} (and hence Theorem~\ref{thm_main} for finite-dimensional $B$ and $C$), we need to argue that this same map $\bar \cR$ remains valid when we consider any separable space $A$. In order to do this, observe that any separable Hilbert space $A$ can be isometrically embedded into $\bar{A}$~\cite[Theorem~II.7]{simon_book}. To conclude, it suffices to remark that $\Delta_{\bar{\cR}}$ is invariant under isometries applied on the space $A$.

\section{Extension to infinite dimensions} \label{sec_infiniteSystems}
In this section we show how to obtain the statement of Theorem~\ref{thm_main} for separable (not necessarily finite-dimensional) Hilbert spaces $A$, $B$, $C$ from the finite-dimensional case that has been proven in Section~\ref{sec_finiteDim}. For trace non-increasing completely positive maps $\cR_{B \to BC}$ we define the function family
\begin{align} \label{eq_deltaInfinite}
\bar \Delta_{\cR}: \, \,\, \,& \D(A\otimes B \otimes C) \to  \R\cup \{-\infty\}  \nonumber \\
&\rho_{ABC} \mapsto F\bigl(\rho_{ABC},\cR_{B \to BC}(\rho_{AB}) \bigr) - 2^{-\frac{1}{2} I(A:C|B)_{\rho}} \ ,
\end{align}
where $ \D(A\otimes B \otimes C)$ denotes the set of states on $A \otimes B \otimes C$.
We will use the same notation as introduced at the beginning of Section~\ref{sec_finiteDim}. 
In addition, we take $\cS$ to be the set of all  states on $A \otimes B \otimes C$ with a fixed marginal $\rho_{BC}$ on $B \otimes C$. The proof proceeds in two steps where we first show that there exists a sequence of recovery maps $\{\cR_{B \to BC}^{k}\}_{k \in \N}$ such that $\lim_{k \to \infty} \bar \Delta_{\cR^{k}}(\cS) \geq 0$, where the property that all elements of $\cS$ have the same marginal on the $B \otimes C$ system will be important.
In the second step we conclude by an approximation argument that there exists a recovery map $\cR_{B \to BC}$ such that $\bar \Delta_{\cR}(\cS) \geq 0$.

\subsection*{\hypertarget{step_4_1}{Step 1: }Existence of a sequence of recovery maps} 
We start by introducing some notation that is used within this step.
Let $\{\Pi_B^{b}\}_{b \in \N}$ and $\{\Pi_C^c\}_{c \in \N}$ be  sequences of finite-rank projectors on $B$ and $C$ which converge to $\id_B$ and $\id_C$ with respect to the weak operator topology. For any given $\rho_{ABC} \in \D(A \otimes B \otimes C)$ consider the normalized projected states
\begin{align} \label{eq_rho_bc}
  \rho_{A B C}^{b,c} := \frac{(\id_A \otimes \Pi^b_{B} \otimes \Pi^c_{C}) \rho_{A B C} (\id_A \otimes \Pi_{B}^{b} \otimes \Pi_{C}^c)}{\tr\bigl((\id_A \otimes \Pi_{B}^{b} \otimes \Pi_{C}^c) \rho_{A B C}\bigr)} \ 
\end{align}
and
\begin{align}
  \rho_{A B C}^{c} := \frac{(\id_A \otimes \id_{B} \otimes \Pi^c_{C}) \rho_{A B C} (\id_A \otimes \id_{B} \otimes \Pi_{C}^c)}{\tr\bigl((\id_A \otimes \id_{B}\otimes \Pi_{C}^c) \rho_{A B C}\bigr)} \ ,
\end{align}
where for any $c \in \mathbb{N}$, the sequence $\{\rho_{A B C}^{b, c}\}_{b \in \mathbb{N}}$ converges to $\rho^{c}_{A B C}$ in the trace norm (see, e.g., Corollary~2 of~\cite{FAR11}) and the sequence $\{\rho_{A B C}^{c}\}_{c \in \mathbb{N}}$ converges to $\rho_{A B C}$ also in the trace norm. 
Let $\cS^{b,c}$ be the set of states that is generated by~\eqref{eq_rho_bc} for all $\rho_{ABC} \in \cS$. We note that for any given $b$, $c$ all elements of $\cS^{b,c}$ have an identical marginal on $B \otimes C$.
Let $\cR^{b,c}_{B \to BC}$ denote a recovery map that satisfies $\bar \Delta_{\cR^{b,c}}(\cS^{b,c}) \geq 0$ whose existence is established in the proof of Theorem~\ref{thm_main} for finite-dimensional systems $B$ and $C$ (see Section~\ref{sec_finiteDim}).
We next state a lemma that explains how $\bar \Delta_{\cR}(\rho)$ changes when we replace $\rho$ by a projected state $\rho^{b,c}$.

\begin{lemma} \label{lem_DeltaProjector}
For any $\rho_{BC} \in \D(B \otimes C)$  there exists a sequence of reals $\{\xi^{b,c} \}_{b,c \in \N}$ with\footnote{The precise form of the sequence $\{\xi^{b,c} \}_{b,c \in \N}$ is given in the proof (see Equation~\eqref{eq_defXi}).} $\lim_{c \to \infty} \lim_{b \to \infty} \xi^{b,c}=0$, such that for any $\cR \in \TPCP(B,B \otimes C)$, any extension $\rho_{ABC}$ of $\rho_{BC}$, and $\rho_{ABC}^{b,c}$ as given in~\eqref{eq_rho_bc} we have
\begin{align}
\bar \Delta_{\cR}(\rho^{b,c}) - \bar \Delta_{\cR}(\rho) \leq \xi^{b,c} \quad \textnormal{for all} \quad b,c \in \N \ .
\end{align}
\end{lemma}
\begin{proof}
We note that local projections applied to the subsystem $C$ can only decrease the mutual information, i.e.,
\begin{align} \label{eq_MIdecreaseLoc}
\tr(\Pi^c_C \rho_{C})I(A:C|B)_{\rho^c}  \leq  I(A:C|B)_{\rho} \ .
\end{align}
The Alicki-Fannes inequality~\cite{AF03} ensures that for a fixed finite-dimensional system $C$ the conditional mutual information $I(A:C|B)_{\rho} = H(C|B)_{\rho}-H(C|AB)_{\rho}$ is continuous in $\rho$ with respect to the trace norm, i.e.,
\begin{align}  \label{eq_cmisec}
I(A:C|B)_{\rho^{b,c}} \leq I(A:C|B)_{\rho^c} + 8 \eps^{b,c} \log(\rank\, \Pi^c_C) + 4 h(\eps^{b,c}) \ ,
\end{align}
where $\eps^{b,c}=\| \rho^{b,c}_{ABC} - \rho^c_{ABC}\|_1$ and $h(\cdot)$ denotes the binary Shannon entropy function defined by $h(p):=-p \log_2(p)-(1-p)\log_2(1-p)$ for $0\leq p \leq 1$. 
Using the Fuchs-van de Graaf inequality~\cite{fuchs99} and Lemma~\ref{lem_varGentleMeas}, we find
\begin{align} \label{eq_epsiUB}
\eps^{b,c}\leq 2 \sqrt{1-F(\rho^{b,c}_{ABC},\rho^c_{ABC})^2}
\leq 2 \sqrt{1 - \tr(\Pi^b_{B}\otimes \Pi^c_C \rho_{BC})/\tr(\Pi^c_C \rho_C)}\ .
\end{align}
Combining~\eqref{eq_MIdecreaseLoc} and~\eqref{eq_cmisec} yields
\begin{align} \label{eq_endCMI}
 I(A:C|B)_{\rho^{b,c}} \leq \frac{1}{\tr(\Pi^c_C \rho_C)}I(A:C|B)_{\rho}+ 8 \eps^{b,c} \log(\rank\, \Pi^c_C) + 4 h(\eps^{b,c})  \ .
\end{align}
Since $x^y \leq x-y +1$ for $x,y \in [0,1]$,\footnote{For $x=0$ the statement clearly holds. For $(0,1]\times [0,1] \ni (x,y) \mapsto f(x,y):=x^y-x+y-1 \in \R$ we find by using the convexity of $y \mapsto f(x,y)$ that $\max_{x \in (0,1]} \max_{y \in [0,1]} f(x,y) = 0$.
} we find
\begin{align} \label{eq_stepinbetween}
2^{-\frac{1}{2}I(A:C|B)_{\rho}} - 2^{-\frac{1}{2}I(A:C|B)_{\rho^{b,c}}} & \leq 2^{-\frac{1}{2}I(A:C|B)_{\rho}} - 2^{-\frac{1}{2} \tr(\Pi^c_C \rho_C )I(A:C|B)_{\rho^{b,c}}} - \tr(\Pi^c_C \rho_C) +1 \ .
\end{align}
According to~\eqref{eq_endCMI} and since $2^{-x} \geq 1 - \ln(2) x$ for $x \in \R$, we have
\begin{align} 
2^{-\frac{1}{2} \tr(\Pi^c_C \rho_C) I(A:C|B)_{\rho^{b,c}}} &\geq 2^{-\frac{1}{2}I(A:C|B)_{\rho}} 2^{-\frac{1}{2}\tr(\Pi^c_C \rho_C)(8 \eps^{b,c} \log(\rank\, \Pi^c_C) + 4 h(\eps^{b,c}))} \nonumber \\
&\geq 2^{-\frac{1}{2}I(A:C|B)_{\rho}} - \frac{\ln(2)}{2}\tr(\Pi^c_C \rho_C)  \bigl(8 \eps^{b,c} \log(\rank\, \Pi^c_C) + 4 h(\eps^{b,c})\bigr) \ . \label{eq_juststep}
\end{align}
Combining~\eqref{eq_stepinbetween} and~\eqref{eq_juststep} yields
\begin{align}
2^{-\frac{1}{2}I(A:C|B)_{\rho}} - 2^{-\frac{1}{2}I(A:C|B)_{\rho^{b,c}}} &\leq \frac{\ln(2)}{2} \tr(\Pi^c_C \rho_C) \bigl(8 \eps^{b,c} \log(\rank\, \Pi^c_C) + 4 h(\eps^{b,c}) \bigr) +\bigl(1 - \tr(\Pi^c_C \rho_C)\bigr) \nonumber\\
&\leq \frac{\ln(2)}{2} \bigl(8 \eps^{b,c} \log(\rank\, \Pi^c_C) + 4 h(\eps^{b,c}) \bigr) +\bigl(1 - \tr(\Pi^c_C \rho_C)\bigr)   \ . \label{eq_doneDSS}
\end{align}

For two states $\sigma_1$ and $\sigma_2$ let $P(\sigma_1,\sigma_2):=\sqrt{1-F(\sigma_1,\sigma_2)^2}$ denote the purified distance. Applying the Fuchs-van de Graaf inequality~\cite{fuchs99} and Lemma~\ref{lem_varGentleMeas} gives
\begin{align}
P(\rho_{ABC}, \rho^{b,c}_{ABC})^2 
= 1-F(\rho_{ABC}, \rho^{b,c}_{ABC})^2
\leq 1- \tr(\Pi^b_{B} \otimes \Pi^c_C  \, \rho_{BC} ) \ .  \label{eq_gentle1}
\end{align} 
Since the purified distance is a metric~\cite{TCR10} that is monotonous under trace-preserving completely positive maps~\cite[Theorem~3.4]{tomamichel_phd},~\eqref{eq_gentle1} gives
\begin{align}
&P\bigl(\rho_{ABC},\cR_{B\to BC}(\rho_{AB})\bigr) \nonumber\\
&\hspace{20mm}\leq P(\rho_{ABC}, \rho^{b,c}_{ABC}) + P\bigl(\rho^{b,c}_{ABC},\cR_{B\to BC}(\rho^{b,c}_{AB})\bigr)+ P\bigl(\cR_{B\to BC}( \rho^{b,c}_{AB} ),\cR_{B\to BC}(\rho_{AB})\bigr)  \nonumber\\
 &\hspace{20mm}\leq 2 P(\rho_{ABC}, \rho^{b,c}_{ABC} ) + P\bigl(\rho^{b,c}_{ABC} ,\cR_{B\to BC}(\rho^{b,c}_{AB})\bigr)\nonumber \\
 &\hspace{20mm}\leq P\bigl( \rho^{b,c}_{ABC} ,\cR_{B\to BC}(\rho^{b,c}_{AB})\bigr) + 2 \sqrt{1- \tr(\Pi^b_{B} \otimes \Pi^c_C \, \rho_{BC}) } \ . \label{eq_purifMid}
\end{align}
As the fidelity for states lies between zero and one, \eqref{eq_purifMid} implies 
\begin{align}
&F\bigl(\rho^{b,c}_{ABC} ,\cR_{B\to BC}(\rho^{b,c}_{AB})\bigr)^2 \nonumber \\
&\hspace{20mm}\leq F\bigl(\rho_{ABC},\cR_{B\to BC}(\rho_{AB})\bigr)^2 + 4 \bigl(1-\tr(\Pi^b_{B} \otimes \Pi^c_C \rho_{BC}) \bigr) + 4 \sqrt{1-\tr(\Pi^b_{B} \otimes \Pi^c_C \rho_{BC}) }\nonumber\\
&\hspace{20mm}\leq F\bigl(\rho_{ABC},\cR_{B\to BC}(\rho_{AB})\bigr)^2 + 8\sqrt{1-\tr(\Pi^b_{B} \otimes \Pi^c_C \rho_{BC})} \nonumber\\
&\hspace{20mm}\leq \Bigl(F\bigl(\rho_{ABC},\cR_{B\to BC}(\rho_{AB})\bigr) + 2 \sqrt{2} \bigl(1-\tr(\Pi^b_{B} \otimes \Pi^c_C \rho_{BC})\bigr)^{\frac{1}{4}} \Bigr)^2 \ .
\label{eq_ddddone}
\end{align} 
This implies that
\begin{align} \label{eq_Fend}
F\bigl(\rho^{b,c}_{ABC},\cR_{B \to BC}(\rho^{b,c}_{AB})\bigr) \leq F\bigl( \rho_{ABC},\cR_{B \to BC}( \rho_{AB})\bigr) +2 \sqrt{2} \bigl(1-\tr(\Pi^b_{B} \otimes \Pi^c_C \rho_{BC})\bigr)^{\frac{1}{4}} \ .
\end{align}

By definition of the quantity $\bar\Delta_{\cR}(\cdot)$ (see Equation~\eqref{eq_deltaInfinite}) the combination of~\eqref{eq_doneDSS} and~\eqref{eq_Fend} yields
\begin{align}
&\bar \Delta_{\cR}(\rho^{b,c}) - \bar \Delta_{\cR}(\rho) \nonumber \\
&\hspace{5mm}\leq \frac{\ln(2)}{2} \bigl(8 \eps^{b,c} \log(\rank\, \Pi^c_C)  + 4 h(\eps^{b,c}) \bigr) +\bigl(1 - \tr(\Pi^c_C \rho_C) \bigr) + 2 \sqrt{2} \bigl(1- \tr(\Pi^b_B \otimes \Pi^c_C \rho_{BC}) \bigr)^{\frac{1}{4}} \nonumber \\
&\hspace{5mm}=:\xi^{b,c} \ ,  \label{eq_defXi}
\end{align}
where $\eps^{b,c}$ is bounded by~\eqref{eq_epsiUB}.
By Lemma~\ref{lem_convergenceWeak}, we find $\lim_{b \to \infty} \tr(\Pi_B^b \otimes \Pi_C^c \rho_{BC}) =\tr(\Pi_C^c \rho_C)$ for all $c \in \N$ and hence $\lim_{b \to \infty} \eps^{b,c}=0$ for any $c\in \N$. Furthermore, we have $\lim_{c \to \infty} \tr(\Pi_C^c \rho_C) =1$ and $\lim_{c \to \infty} \lim_{b \to \infty} \tr(\Pi_B^b \otimes \Pi_C^c \rho_{BC}) =1$ which implies that $\lim_{c \to \infty} \lim_{b \to \infty} \xi^{b,c} = 0$. This proves the assertion.
%
%
\end{proof}

By Lemma~\ref{lem_DeltaProjector}, using the notation defined at the beginning of Step~\hyperlink{step_4_1}{1}, we find \begin{multline} \label{eq_seqneceImportantStep}
\limsup_{c\to \infty} \limsup_{b \to \infty}  \bar \Delta_{\cR^{b,c}}(\cS) 
=\limsup_{c\to \infty} \limsup_{b \to \infty} \inf_{\rho \in \cS} \bar \Delta_{\cR^{b,c}}(\rho) 
\geq \limsup_{c\to \infty}  \limsup_{b \to \infty}  \inf_{\rho \in \cS} \bigl\{\bar \Delta_{\cR^{b,c}}(\rho^{b,c}) -\xi^{b,c} \bigr\} \\
=\limsup_{c\to \infty}  \limsup_{b \to \infty} \bigl\{ \inf_{\rho^{b,c} \in \cS^{b,c}} \bar \Delta_{\cR^{b,c}}(\rho^{b,c})\bigr\}  -\xi^{b,c} 
=\limsup_{c\to \infty}  \limsup_{b \to \infty}  \bar \Delta_{\cR^{b,c}}(\cS^{b,c}) 
\geq 0 \ ,
\end{multline}
where the second equality step is valid since all states in $\cS$ have the same fixed marginal on ${B \otimes C}$ and since the sequence $\{\xi^{b,c}\}_{b,c \in\N}$ only depends on this marginal.
The penultimate step uses that $\lim_{c \to \infty} \lim_{b \to \infty} \xi^{b,c}=0$.
The final inequality follows by definition of $\cR_{B \to BC}^{b,c}$.
Inequality~\eqref{eq_seqneceImportantStep} implies that there exist sequences $\{b_k\}_{k\in \N}$ and $\{c_k\}_{k \in \N}$ such that $\limsup_{k\to \infty}\Delta_{\cR^{b_k,c_k}}(\cS) \geq 0$. Setting $\cR_{B \to BC}^k = \cR_{B \to BC}^{b_k,c_k}$ then implies that there exists a sequence $\{\cR_{B \to BC}^{k}\}_{k \in \mathbb{N}}$ of recovery maps that satisfies
\begin{align} \label{eq_startingpoint}
  \limsup_{k \to \infty} \bar \Delta_{\cR^{k}}(\cS) \geq 0 \ .
\end{align}


\subsection*{\hypertarget{step_4_2}{Step 2: }Existence of a limit}
Recall that $\cS$ is the set of density operators on $A \otimes B \otimes C$ with a fixed marginal $\rho_{B C}$ on $B \otimes C$. 
The goal of this step is to use~\eqref{eq_startingpoint} to prove that there exists a recovery map $\cR_{B \to BC}$ such that
\begin{align} \label{eq_Rtoprove}
  \bar \Delta_{\cR}(\cS) \geq 0 \ .
\end{align}

Let $\{\Pi_B^m\}_{m \in \mathbb{N}}$ and $\{\Pi_C^m\}_{m \in \mathbb{N}}$ be sequences of projectors with rank $m$ that weakly converge to $\id_B$ and $\id_C$, respectively. Furthermore, for any $m$ and any $\cR \in \TPCP(B,B\otimes C)$ let $[\cR]^m$ be the trace non-increasing map obtained from $\cR$ by projecting the input and output with $\Pi_B^m$ and $\Pi_B^m \otimes \Pi_C^m$, respectively.  We start with a preparatory lemma that proves a relation between $\bar \Delta_{[\cR]^m}(\cS)$ and $\bar \Delta_{\cR}(\cS)$.
\begin{lemma} \label{lem_DeltaProj}
For any $\rho_{BC} \in \D(B \otimes C)$ there exists a sequence of reals $\{\delta^m\}_{m \in \N}$ with $\lim_{m \to \infty} \delta^m = 0$,\footnote{The precise form of the sequence $\{\delta^m\}_{m \in \N}$ can be found in the proof (see Equation~\eqref{eq_deltaLem}).} such that for any $\cR\in \TPCP(B,B\otimes C)$ we have
\begin{align} \label{eq_GMLimplic}
  \bar \Delta_{[\cR]^m}(\cS) \geq   \bar \Delta_{\cR}(\cS) - \delta^m - 4 \,  \eps^{\frac{1}{4}}  \ ,
\end{align}
where $\norm{\cR(\rho_B)-\rho_{BC}}_1 \leq \eps$.
\end{lemma}
\begin{proof}
For any $\rho_{ABC} \in \cS$ and any $m \in \N$ let us define the non-negative operator $\hat \rho^m_{AB}:= {(\id_A \otimes \Pi_B^m)}\rho_{AB}$ ${(\id_A \otimes \Pi_B^m)}$.
By definition of $\bar \Delta_{\cR}(\cdot)$ (see Equation~\eqref{eq_deltaInfinite}), it suffices to show that for any $\rho_{ABC} \in \cS$, any $\cR \in \TPCP(B,B \otimes C)$, $\eps \in [0,2]$ such that $\norm{\cR(\rho_B) - \rho_{BC}}_1 \leq \eps$ and 
\begin{align}
 \tilde \rho^m_{ABC}:= (\id_A \otimes \Pi_B^m \otimes \Pi_C^m) \cR_{B \to BC}(\hat \rho^m_{AB} ) (\id_A \otimes \Pi_B^m \otimes \Pi_C^m)
\end{align}
we have $F\bigl(\rho_{ABC}, \tilde \rho^m_{ABC} \bigr) \geq F\bigl(\rho_{ABC},\cR_{B \to BC}(\rho_{AB}) \bigr) - \delta^m -  4\, \eps^{\frac{1}{4}}$. As in Step~\hyperlink{step_4_1}{1} let $P(\cdot,\cdot)$ denote the purified distance. Lemma~\ref{lem_varGentleMeas} implies that
\begin{align} \label{eq_firstPartGM}
P(\rho_{AB},\hat \rho^m_{AB})^2 = 1- F(\rho_{AB}, \hat \rho^m_{AB})^2 \leq  1 - \tr(\rho_B \Pi_B^m)^2 \ .
\end{align}
Similarly, we obtain 
\begin{align} \label{eq_inthemiddlestep}
P\bigl(\cR_{B \to BC}(\hat \rho^m_{AB}), \tilde \rho^m_{ABC} \bigr)^2 &\leq 1- \tr\bigl(\cR_{B \to BC}(\hat \rho^m_{AB}) \Pi_B^m \otimes \Pi_C^m\bigr)^2 = 1 - \tr\bigl( \cR_{B \to BC}(\hat \rho^m_B) \Pi_B^m \otimes \Pi_C^m \bigr)^2 \ .
\end{align}
By H\"older's inequality, monotonicity of the trace norm for trace-preserving completely positive maps~\cite[Example~9.1.8 and Corollary~9.1.10]{wilde_book} and~\eqref{eq_firstPartGM} together with the Fuchs-van de Graaf inequality~\cite{fuchs99} and Lemma~\ref{lem_varGentleMeas} we find
\begin{multline} \label{eq_holderstep}
\left| \tr \Bigl( \bigl(\cR_{B \to BC}(\hat \rho^m_B) - \cR_{B \to BC}(\rho_B) \bigr) \Pi_B^m \otimes \Pi_C^m \Bigr) \right| \leq \norm{\cR_{B \to BC}(\rho_B)-\cR_{B \to BC}(\hat \rho^m_B)}_1 \norm{\Pi_B^m \otimes \Pi_C^m}_{\infty} \\
= \norm{\cR_{B \to BC}(\rho_B)-\cR_{B \to BC}(\hat \rho^m_B)}_1 \leq \norm{\rho_B-\hat \rho^m_B}_1\leq \norm{\rho_{AB}-\hat \rho^m_{AB}}_1 \leq  2\sqrt{1 - \tr(\rho_B \Pi_B^m)^2} \ .
\end{multline}
Combining~\eqref{eq_inthemiddlestep},~\eqref{eq_holderstep} and H\"older's inequality together with the assumption $\norm{\cR(\rho_B) - \rho_{BC}}_1 \leq \eps$ gives
\begin{align}
P\bigl(\cR_{B \to BC}(\hat \rho^m_{AB}),  \tilde \rho^m_{ABC} \bigr)^2 
&\leq 1 - \tr\bigl( \cR_{B \to BC}(\rho_B) \Pi_B^m \otimes \Pi_C^m \bigr)^2 +4 \sqrt{1 - \tr(\rho_B \Pi_B^m)^2}   \nonumber \\
&\leq 1 - \tr(\rho_{BC} \Pi_B^m \otimes \Pi_C^m)^2  + 4 \sqrt{1 - \tr(\rho_B \Pi_B^m)^2}+ 2\eps\ . \label{eq_secondPartGM}
\end{align}

Inequalities~\eqref{eq_firstPartGM},~\eqref{eq_secondPartGM} and the monotonicity of the purified distance under trace-preserving and completely positive maps~\cite[Theorem~3.4]{tomamichel_phd} show that 
\begin{align}
&P(\rho_{ABC}, \tilde \rho^m_{ABC}) \nonumber \\
&\hspace{4mm}\leq P\bigl(\rho_{ABC},\cR_{B \to BC}(\rho_{AB})\bigr) + P\bigl(\cR_{B\to BC}(\rho_{AB}),\cR_{B \to BC}(\hat \rho^m_{AB})\bigr) +P\bigl(\cR_{B\to BC}(\hat \rho^m_{AB}), \tilde\rho^m_{ABC}\bigr) \nonumber\\
&\hspace{4mm}\leq P\bigl(\rho_{ABC},\cR_{B \to BC}(\rho_{AB})\bigr) + P\bigl(\rho_{AB},\hat \rho^m_{AB}\bigr) +P\bigl(\cR_{B\to BC}(\hat \rho^m_{AB}), \tilde \rho^m_{ABC}\bigr) \nonumber\\
&\hspace{4mm}\leq P\bigl(\rho_{ABC},\cR_{B \to BC}(\rho_{AB})\bigr) + (\delta^m)^2/8 + \sqrt{2 \eps}   \ , \label{eq_impStepMid}
\end{align}
for 
\begin{align} \label{eq_deltaLem}
\delta^m := \sqrt{8}\left(\sqrt{1 - \tr(\rho_B \Pi_B^m)^2} + \sqrt{1 - \tr(\rho_{BC} \Pi_B^m \otimes \Pi_C^m)^2  + 4 \sqrt{1 - \tr(\rho_B \Pi_B^m)^2}}\right)^{\frac{1}{2}} \ .
\end{align}
As the purified distance between two states lies inside the interval $[0,1]$ and since $(\delta^m)^2/8 + \sqrt{2 \eps} \in [0,6]$,~\eqref{eq_impStepMid} implies that whenever $F(\rho_{ABC},\cR_{B \to BC}(\rho_{AB}))^2 \geq (\delta^m)^2 + 8 \sqrt{2 \eps}$, we have
\begin{multline}
F(\rho_{ABC}, \tilde \rho^m_{ABC})^2 \geq F\bigl(\rho_{ABC},\cR_{B \to BC}(\rho_{AB}) \bigr)^2 - (\delta^m)^2 - 8 \sqrt{2 \eps} \\
\geq \Bigl(F\bigl(\rho_{ABC},\cR_{B \to BC}(\rho_{AB}) \bigr) - \sqrt{(\delta^m)^2 + 8 \sqrt{2 \eps}} \Bigr)^2 \ .
\end{multline}
As a result, we find
\begin{align}
F(\rho_{ABC}, \tilde \rho^m_{ABC}) \geq F\bigl(\rho_{ABC},\cR_{B \to BC}(\rho_{AB}) \bigr) - \delta^m - \sqrt{8} (2 \eps)^{\frac{1}{4}} \ ,
\end{align}
which proves~\eqref{eq_GMLimplic} since $\sqrt{8} \, 2^{\frac{1}{4}} \leq 4$. 

Recall that $B$ and $C$ are separable Hilbert spaces and that $\{\Pi_B^m\}_{m\in \N}$ and $\{\Pi_{B}^{m}\otimes \Pi_{C}^{m}\}_{m\in \N}$ converge weakly to $\id_{B}$ and $\id_B \otimes \id_C$ respectively. Lemma~\ref{lem_convergenceWeak} thus shows that $\lim_{m \to \infty} \delta^m =0$ since $\lim_{m\to \infty} \tr(\rho_B \, \Pi_{B}^{m}) =1$ and $\lim_{m \to \infty} \tr(\rho_{BC} {\Pi_{B}^{m}\otimes \Pi_{C}^{m}}) =1$. 
\end{proof}

The following lemma proves that for sufficiently large $m$ and a recovery map $\cR_{B \to BC}$ that maps $\rho_B$ to density operators that are close to $\rho_{B C}$, the operator $[\cR]^m(\rho_{AB})$ has a trace that is bounded from below by essentially one. 
\begin{lemma} \label{lem_uniformConv}
Let $A$, $B$, and $C$ be separable Hilbert spaces. For any density operator $\rho_{AB} \in \D(A \otimes B)$ and any $\cR \in \TPCP(B,B \otimes C)$ we have
\begin{align}
\tr \bigl([\cR]^m(\rho_{AB}) \bigr) \geq \tr(\Pi_B^m \otimes \Pi_C^m \rho_{BC}) - 2 \sqrt{1-\tr(\Pi_B^m \rho_B)} - \norm{\cR(\rho_B)- \rho_{BC}}_1 \ .
\end{align}
\end{lemma}
\begin{proof}
We first note that by H\"older's inequality and monotonicity of the trace norm for trace-preserving completely positive maps~\cite[Example~9.1.8 and Corollary~9.1.10]{wilde_book} we have
\begin{align}
\left| \tr\left(\Pi^m_B \otimes \Pi^m_C \bigl(\cR(\rho_B) - \cR(\Pi_B^m \rho_B \Pi_B^m) \bigr) \right) \right| \leq \norm{\cR(\rho_B) - \cR(\Pi_B^m \rho_B \Pi_B^m)}_1 \leq \norm{\rho_B- \Pi_B^m \rho_B \Pi_B^m}_1 \ .
\end{align}
Together with H\"older's inequality this implies
\begin{align}
\tr \bigl([\cR]^m(\rho_{AB}) \bigr) 
&= \tr\bigl(\Pi_B^m\otimes \Pi_C^m \, \cR(\Pi_B^m \rho_{AB} \Pi_B^m) \bigr) 
= \tr\bigl(\Pi_B^m\otimes \Pi_C^m \, \cR(\Pi_B^m \rho_{B} \Pi_B^m) \bigr) \nonumber\\
 &\geq \tr\bigl(\Pi_B^m\otimes \Pi_C^m \, \cR(\rho_{B}) \bigr) - \norm{\rho_B- \Pi_B^m \rho_B \Pi_B^m}_1 \nonumber \\
  &\geq \tr\bigl(\Pi_B^m \otimes \Pi_C^m\rho_{BC} \bigr) - \norm{\rho_B- \Pi_B^m \rho_B \Pi_B^m}_1 - \norm{\cR(\rho_B)- \rho_{BC}}_1 \ . \label{eq_lemOmarstep}
\end{align}
Combining Lemma~\ref{lem_tracedistancefidelity} and Lemma~\ref{lem_varGentleMeas} gives
\begin{multline}
\norm{\rho_B - \Pi_B^m \rho_B \Pi_B^m}_1 
\leq 2 \sqrt{1 - F(\rho_B, \Pi_B^m \rho_B \Pi_B^m)^2} 
= 2 \sqrt{1-\tr(\Pi_B^m \rho_B) F\bigl(\rho_B, \Pi_B^m \rho_B \Pi_B^m/\tr(\Pi_B^m \rho_B)\bigr)^2}\\
\leq 2 \sqrt{1-\tr(\Pi_B^m \rho_B)} \ ,
\end{multline}
which together with~\eqref{eq_lemOmarstep} proves the assertion.
\end{proof}

According to~\eqref{eq_startingpoint} the mappings $\cR^{k}$ satisfy 
\begin{align} \label{eq_specifyEpsK}
\bar \Delta_{\cR^k}(\cS) \geq - \tilde \eps^k \ ,
\end{align}
with $\tilde \eps^k \geq 0$ such that $\liminf_{k \to \infty} \tilde \eps^k =0$. As explained in Remark~\ref{rmk_BtoBC}, by considering a state $\bar \rho_{ABC} = \rho_A \otimes \rho_{BC} \in \cS$, \eqref{eq_specifyEpsK} implies $F(\rho_{BC}, \cR^k(\rho_B)) \geq - \tilde \eps^k +1$. Applying the Fuchs-van de Graaf inequality~\cite{fuchs99} gives
\begin{align} \label{eq_RkBalmostToBC}
\norm{\rho_{BC} - \cR^k(\rho_B)}_1 \leq 2 \sqrt{\tilde \eps^k (2-\tilde \eps^k )} =: \eps^k \ ,
\end{align}
where $\liminf_{k \to \infty} \eps^k = 0$ because $\liminf_{k \to \infty} \tilde \eps^k = 0$.

By Lemma~\ref{lem_DeltaProj} we have
\begin{align}
  \bar \Delta_{[\cR^{k}]^m}(\cS) \geq   \bar \Delta_{\cR^{k}}(\cS) - 4 (\varepsilon^k)^{\frac{1}{4}} - \delta^m  \ .
\end{align}
Hence, using our starting point~\eqref{eq_startingpoint} above,
\begin{align}
  \limsup_{k \to \infty}  \bar \Delta_{[\cR^{k}]^m}(\cS)  \geq \limsup_{k \to \infty} \bar \Delta_{\cR^{k}}(\cS)  -4(\varepsilon^k)^{\frac{1}{4}}  - \delta^m \geq - \delta^m \ .
\end{align}
Because, for any fixed $m\in \N$, the mappings $[\cR^{k}]^m$, for $k \in \mathbb{N}$, are all contained in the same finite-dimensional subspace (i.e., the set of trace non-increasing maps from  operators on the support of $\Pi^m_B$ to operators on the support of  $\Pi^m_B \otimes \Pi^m_C$), and because the space of all such mappings is compact (see Remark~\ref{rmk_TPCPcompact}), for any fixed $m \in \N$ there exists a subsequence of the sequence $\{[\cR^{k}]^m\}_{k \in \N}$ that converges. Specifically for any fixed $m \in \N$  there exists a sequence $\{k^m_i\}_{i \in \mathbb{N}}$ such that 
\begin{align}
  \bar{\cR}^{m} := \lim_{i \to \infty} [\cR^{k^m_i}]^m 
\end{align}
is well defined. Furthermore, because of the continuity of $\cR \mapsto \bar \Delta_{\cR}(\rho_{ABC})$ on the set of maps from  operators on the support of $\Pi^m_B$ to operators on the support of  $\Pi^m_B \otimes \Pi^m_C$ (see Lemma~\ref{lem_contInR}), we have
\begin{multline}
  \bar \Delta_{\bar{\cR}^{m}}(\cS)
  = \inf_{\rho \in \cS} \bar \Delta_{\bar{\cR}^{m}}(\rho) 
  =  \inf_{\rho \in \cS} \lim_{i \to \infty} \bar \Delta_{[\cR^{k^m_i}]^m}(\rho)  
  \geq \limsup_{i \to \infty} \inf_{\rho \in \cS}  \bar \Delta_{[\cR^{k^m_i}]^m}(\rho) \\
  = \limsup_{i \to \infty}  \bar \Delta_{[\cR^{k^m_i}]^m}(\cS) 
  \geq -\delta^m \ ,
\end{multline}
and, hence,
\begin{align} \label{eq_almostThere}
  \liminf_{m \to \infty} \bar \Delta_{\bar{\cR}^{m}}(\cS) \geq 0 \ .
\end{align}

For any $m \in \mathbb{N}$, let $\rho^{m}_{B C : \bar{B}}$ be the operator obtained by applying $\bar{\cR}^{m}$ to a purification $\rho_{B : \bar{B}}$ of $\rho_B$. Without loss of generality we can assume that the projector $\Pi_B^m$ is in the eigenbasis of $\rho_B$.
Let $\{k_i^{m+1}\}_{i\in \N}$ be a subsequence of $\{k_{i}^{m}\}_{i \in \N}$.
Using the definition of $\bar{\cR}^{m}$ and that $\Pi_B^m \leq \Pi_B^{m'}$, $\Pi_C^m \leq \Pi_C^{m'}$, and $\Pi_{\bar B}^m \leq \Pi_{\bar B}^{m'}$ for $m \leq m'$, we obtain
\begin{align}
\rho^{m}_{B C : \bar{B}} 
&= \bar{\cR}^m(\rho_{B : \bar{B}}) 
= \lim_{i \to \infty} [\cR^{k_{i}^{m+1}}]^{m}(\rho_{B : \bar B}) 
= \lim_{i \to \infty} (\Pi_{B}^{m}\otimes \Pi_C^m)[\cR^{k_{i}^{m+1}}]^{m+1}(\Pi_{B}^{m} \rho_{B : \bar B} \Pi_B^m) (\Pi_{B}^{m}\otimes \Pi_{C}^m) \nonumber\\
&= \lim_{i \to \infty} (\Pi_{B}^{m}\otimes \Pi_C^m \otimes \Pi_{\bar B}^{m})[\cR^{k_{i}^{m+1}}]^{m+1}( \rho_{B : \bar B} ) (\Pi_{B}^{m}\otimes \Pi_C^m \otimes \Pi_{\bar B}^{m}) \nonumber \\
&= (\Pi_{B}^m\otimes \Pi_C^m\otimes \Pi_{\bar B}^{m}) \bar{\cR}^{m+1}(\rho_{B : \bar B}) (\Pi_{B}^m\otimes \Pi_C^m \otimes \Pi_{\bar B}^{m}) \nonumber \\
&= (\Pi_{B}^m\otimes \Pi_C^m\otimes \Pi_{\bar B}^{m}) \rho^{m+1}_{BC : \bar B} (\Pi_{B}^m\otimes \Pi_C^m \otimes \Pi_{\bar B}^{m})\ .
\end{align}
As a result, since $\Pi_B^m \leq \Pi_B^{m'}$, $\Pi_C^m \leq \Pi_C^{m'}$, and $\Pi_{\bar B}^m \leq \Pi_{\bar B}^{m'}$ for $m \leq m'$, we have for any $m\leq m'$
\begin{align} \label{eq_mPrimeAndm}
  \rho^{m}_{B C : \bar{B}} = (\Pi^m_B \otimes \Pi^m_C \otimes \Pi^m_{\bar{B}}) \rho^{m'}_{B C : \bar{B}} (\Pi^m_B \otimes \Pi^m_C \otimes \Pi^m_{\bar{B}}) \ .
\end{align}
 Lemma~\ref{lem_varGentleMeas} together with~\eqref{eq_mPrimeAndm} implies
\begin{multline}
F(\rho_{BC:\bar B}^{m}, \rho_{BC:\bar B}^{m'}) 
= F(\Pi_B^m \otimes \Pi_C^m \otimes \Pi_{\bar B}^m \rho_{BC:\bar B}^{m'} \Pi_B^m \otimes \Pi_C^m \otimes \Pi_{\bar B}^m,  \rho_{BC:\bar B}^{m'}) \\
\geq \tr(\rho_{BC:\bar B}^{m'}\Pi_B^m \otimes \Pi_C^m \otimes \Pi_{\bar B}^m) 
= \tr(\rho_{BC:\bar B}^{m}) \ . \label{eq_GMState}
\end{multline}
Lemma~\ref{lem_tracedistancefidelity} yields for $m' \geq m$
\begin{align} \label{eq_CauchyJust}
\norm{\rho_{BC:\bar B}^{m} - \rho_{BC:\bar B}^{m'}}_1 
\leq  2 \sqrt{\tr(\rho_{BC:\bar B}^{m'})^2 - F(\rho_{BC:\bar B}^{m},\rho_{BC:\bar B}^{m'})^2 } \leq 2 \sqrt{\tr(\rho_{BC:\bar B}^{m'})^2 - \tr(\rho_{BC:\bar B}^{m})^2 }  \ .
\end{align}
We now prove that as $m \to \infty$, $\tr(\rho_{BC:\bar B}^{m})$ goes to $1$. Note that since $B$ is a separable Hilbert space and $\rho_{B:\bar B}$ is normalized it can be written as $\rho_{B:\bar B}= \proj{\psi}$, where $\ket{\psi}$ is a state on $B\otimes \bar B$. Furthermore as $\Pi^m_B \otimes \Pi^m_C \otimes \Pi^m_{\bar{B}} \leq \id_{BC\bar{B}}$,~\eqref{eq_mPrimeAndm} implies that 
\begin{align} \label{eq_UBrhom}
\tr(\rho_{BC:\bar B}^{m})\leq \tr(\rho_{BC:\bar B}^{m'})\leq 1 \quad \textnormal{for} \quad m'\geq m \ .
\end{align}
By definition of $\rho_{BC:\bar B}^{m}$, Lemma~\ref{lem_uniformConv} together with~\eqref{eq_RkBalmostToBC} implies that
\begin{align}
\lim_{m \to \infty} \tr(\rho_{BC:\bar B}^{m}) &= \lim_{m \to \infty} \lim_{i \to \infty} \tr\bigl([\cR^{k^m_i}]^m (\rho_{B:\bar B})\bigr) \nonumber\\
&\geq \lim_{m \to \infty} \tr(\Pi_B^m \otimes \Pi_C^m \rho_{BC}) - \norm{\rho_B - \Pi_B^m \rho_B \Pi_B^m}_1 - \liminf_{i \to \infty} \eps^{k^m_i}  \nonumber \\
&\geq \lim_{m \to \infty} \tr(\Pi_B^m \otimes \Pi_C^m \rho_{BC}) - 2 \sqrt{1 - \tr(\Pi_B^m \rho_B)^2}
= 1 \ , \label{eq_weakC} 
\end{align}
where the second inequality uses Lemma~\ref{lem_tracedistancefidelity}, Lemma~\ref{lem_varGentleMeas}, and that $\liminf_{i\to \infty} \eps^{k^m_i} = 0$ for all $m \in \N$. The final step follows by Lemma~\ref{lem_convergenceWeak}.

Equations~\eqref{eq_CauchyJust}, \eqref{eq_UBrhom} and~\eqref{eq_weakC} show that, $\{\rho_{BC:\bar B}^{m}\}_{m \in \mathbb{N}}$ is a Cauchy sequence. Because the set of sub-normalized non-negative operators (i.e.,  the set of sub-normalized density operators) is complete\footnote{We note that the set of sub-normalized density operators on a Hilbert space is clearly closed. Since every Hilbert space is complete and as every closed subspace of a complete space is complete~\cite[Chapter~II, Section~3.4, Proposition~8]{bourbaki_book} this implies that the set of sub-normalized density operators is complete.}, this sequence converges towards such an operator, i.e., we can define a density operator
\begin{align}
  \tilde{\rho}_{B C : \bar{B}} := \lim_{m \to \infty} \rho^{m}_{B C : \bar{B}}  \ .
\end{align}
We note that the operators $\rho_{BC:\bar B}^{m}$ are not normalized in general. However~\eqref{eq_weakC} shows that $\tilde \rho_{BC:\bar B}$ has unit trace.
We now define the recovery map $\cR_{B \to BC}$ as the one that maps $\rho_{B: \bar{B}}$ to $\tilde{\rho}_{B C : \bar{B}}$. We note that this does not uniquely define the recovery map $\cR_{B \to BC}$, which is not a problem as Theorem~\ref{thm_main} proves the \emph{existence} of a recovery map that satisfies~\eqref{eq_main_new_separable} and does not claim that this map is unique.
It remains to show that $\cR_{B \to BC}$ has the property~\eqref{eq_Rtoprove}. This follows from the observation that any density operator $\rho_{A B}$ can be obtained from the purification $\rho_{B : \bar{B}}$ by applying a trace-preserving completely positive map $\cT_{\bar B \to A}$ from $\bar{B}$ to $A$. By Lemma~\ref{lem_contInRho} and because $\cT_{\bar B \to A}$ commutes with any recovery map $\cR_{B \to BC}$ from $B$ to $B \otimes C$, we have 
\begin{align}
  \cR_{B \to BC}(\rho_{A B}) 
  &= (\cR_{B \to BC} \circ \cT_{\bar B \to A})(\rho_{B : \bar{B}}) 
  = (\cT_{\bar B \to A} \circ \cR_{B \to BC})(\rho_{B : \bar{B}}) 
  = \cT_{\bar B \to A}(\tilde{\rho}_{B C : \bar{B}}) \nonumber \\ 
  &=   \cT_{\bar B \to A}(\lim_{m \to \infty} \rho^{m}_{B C : \bar{B}})  
  = \lim_{m \to \infty} \cT_{\bar B \to A}(\rho^{m}_{B C : \bar{B}}) 
  = \lim_{m \to \infty} (\cT_{\bar B \to A} \circ \bar{\cR}_{B \to BC}^{m})(\rho_{B : \bar{B}})  \nonumber \\
  &= \lim_{m \to \infty} (\bar{\cR}_{B \to BC}^{m} \circ \cT_{\bar B \to A})(\rho_{B : \bar{B}}) 
  = \lim_{m \to \infty} \bar{\cR}_{B \to BC}^{m}(\rho_{A B}) \ .
\end{align}
Using the continuity of the fidelity (see, e.g., Lemma~B.9 in~\cite{FR14}), this implies that
\begin{align} \label{eq_essDone}
    \bar \Delta_{\cR}(\rho) =   \lim_{m \to \infty} \bar \Delta_{\bar{\cR}^{m}}(\rho) \ ,
 \end{align}
  for any $\rho \in \cS$. Combining this with~\eqref{eq_almostThere} gives
 \begin{align}
       \bar \Delta_{\cR}(\cS) = \inf_{\rho \in \cS} \bar \Delta_{\cR}(\rho)  =  \inf_{\rho \in \cS} \lim_{m \to \infty}    \bar \Delta_{\bar \cR^{m}}(\rho)
     \geq  \liminf_{m \to \infty} \inf_{\rho \in \cS} \bar \Delta_{\bar{\cR}^{m}}(\rho) =  \liminf_{m \to \infty} \bar \Delta_{\bar{\cR}^{m}}(\cS)  \geq 0 \ ,
\end{align}
which concludes Step~\hyperlink{step_4_2}{2} and thus completes the proof of Theorem~\ref{thm_main} in the general case where $B$ and $C$ are no longer finite-dimensional.


%
%
%
%


\section{Proof of Corollary~\ref{cor_main}} \label{sec_pfCor}
The first statement of Corollary~\ref{cor_main} that holds for separable Hilbert spaces follows immediately from Theorem~\ref{thm_main}, since $2^{-\frac{1}{2}I(A:C|B)_{\rho}} \geq 1 - \frac{\ln(2)}{2}I(A:C|B)_{\rho}$.
The proof of the second statement of Corollary~\ref{cor_main} is partitioned into three steps.\footnote{Although Corollary~\ref{cor_main} does not immediately follow from Theorem~\ref{thm_main} it is justified to term it as such, as it follows by the same proof technique that is used to derive Theorem~\ref{thm_main} (in particular it makes use of Proposition~\ref{prop_finite}).} We first show that a similar method as used in Section~\ref{sec_finiteDim} can be used to reveal certain insights about the structure of the recovery map $\cR_{B \to BC}$ (which is not universal) that satisfies
  \begin{align} \label{eq_prop_main}
   F\bigl(\rho_{A B C}, \cR_{B \to BC}(\rho_{A B})\bigr) \geq 1-\frac{\ln(2)}{2} I(A : C | B)_{\rho}  \ .
  \end{align}
In a second step, by invoking Proposition~\ref{prop_finite}, we use this knowledge to prove that for a fixed $A$ system there exists a recovery map that satisfies~\eqref{eq_prop_main} which is universal and preserves the structure of the non-universal recovery map from before. Finally, in Step~\hyperlink{step_5_3}{3} we show how the dependency on the fixed $A$ system can be removed.

\subsection*{\hypertarget{step_5_1}{Step 1: }Structure of a non-universal recovery map}
We will show that for any density operator $\rho_{ABC}$ on $A \otimes B \otimes C$, where $A$, $B$, and $C$ are finite-dimensional Hilbert spaces there exists a  trace-preserving completely positive map $\cR_{B \to BC}$ that satisfies~\eqref{eq_prop_main} and is of the form
\begin{align} \label{eq_recoveryMapForm}
X_B \mapsto    \rho_{BC}^{\frac{1}{2}}W_{BC} (\rho_B^{-\frac{1}{2}} X_B \rho_B^{-\frac{1}{2}} \otimes \id_C) W_{BC}^{\dagger} \rho_{BC}^{\frac{1}{2}} \ ,
\end{align}
on the support of $\rho_B$, where $W_{BC}$ is a unitary on $B\otimes C$.
We start by proving the following preparatory lemma.

\begin{lemma} \label{lem_removeUnitaries}
For any density operator $\rho_{ABC}$ on $A\otimes B \otimes C$, where $A$, $B$, and $C$ are finite-dimensional Hilbert spaces there exists a trace-preserving completely positive map $\cR_{B \to BC}$ of the form
\begin{align} \label{eq_mapForm}
X_B \mapsto V_{BC}  \rho_{BC}^{\frac{1}{2}}(\rho_B^{-\frac{1}{2}} U_B X_B  U_B^{\dagger}\rho_B^{-\frac{1}{2}} \otimes \id_C) \rho_{BC}^{\frac{1}{2}} V_{BC}^{\dagger} \ ,
\end{align}
where $V_{BC}$ is a unitary on $B\otimes C$ that commutes with $\rho_{BC}$ and $U_B$ is a unitary on $B$ that commutes with $\rho_B$ such that 
  \begin{align} \label{eq_boundM}
   F\bigl(\rho_{A B C}, \cR_{B \to BC}(\rho_{A B})\bigr) \geq 1-\frac{\ln(2)}{2} I(A : C | B)_{\rho}  \ .
  \end{align}
\end{lemma}
\begin{proof}
Let $\rho_{A B C}$ be an arbitrary state on $A\otimes B \otimes C$ and let $\rho^0_{A B C}$ be a Markov chain with the same marginal on the $B\otimes C$ system, i.e., $\rho^0_{B C} = \rho_{B C}$.
For $p \in (0,1]$, define the state
\begin{align}\label{eq_rhoCQ}
  \rho^p_{\hat A A B C} := (1-p) \proj{0}_{\hat A} \otimes \rho^0_{A B C} + p \proj{1}_{\hat A}  \otimes \rho_{A B C} \ .
\end{align}

The main result of~\cite{FR14} (see Theorem~5.1 and Remark~4.3 in~\cite{FR14}) implies that there exists a recovery map $\cR_{B \to B C}$ of the form
\begin{align} \label{eq_mapForm2}
X_B \mapsto V_{BC}  \rho_{BC}^{\frac{1}{2}}(\rho_B^{-\frac{1}{2}} U_B X_B  U_B^{\dagger}\rho_B^{-\frac{1}{2}} \otimes \id_C) \rho_{BC}^{\frac{1}{2}} V_{BC}^{\dagger} \ ,
\end{align}
where $U_B$ is diagonal with respect to the eigenbasis of $\rho_B$, $U_B U_B^\dagger \leq \id_B$ and $V_{BC}$ is a unitary on $B \otimes C$, such that
\begin{align} 
  F\bigl(\rho^p_{\hat A A B C}, \cR_{B \to B C}(\rho^p_{\hat A A B})\bigr) \geq 1- \frac{\ln(2)}{2} I(\hat A A : C | B)_{\rho^p} \ .
\end{align}
By Lemma~\ref{lem_linearizedProp}, using that $I(A:C|B)_{\rho^0}=0$ since $\rho_{ABC}^0$ is a Markov chain, this may be rewritten as
\begin{align} \label{eq_FRboundr}
p \Bigl(1-F\bigl(\rho_{A B C}, \cR_{B \to BC}(\rho_{A B})\bigr)\Bigr) + (1-p) \Bigl(1-F\bigl(\rho^0_{ A B C}, \cR_{B \to BC}(\rho^0_{A B})\bigr)\Bigr) 
 \leq p \frac{\ln(2)}{2} I(A : C | B)_{\rho} \ .
\end{align}

Let us assume by contradiction that any recovery map $\cR_{B \to B C}$ that satisfies~\eqref{eq_FRboundr} does not leave $\rho^0_{A B C}$ invariant, i.e., $\rho^0_{A B C} \neq  \cR_{B \to B C}(\rho^0_{A B})$. This implies that there exists a $\delta_{ \cR} \in (0,1]$, which may depend on the recovery map $\cR_{B \to B C}$, such that $1-F (\rho^0_{A B C}, \cR_{B \to BC}(\rho^0_{A B}))=\delta_{\cR}$. In the following we argue that there exists a universal (i.e., independent of $\cR_{B \to BC}$) constant $\delta \in (0,1]$ such that $1-F(\rho^0_{ A B C},  \cR_{B \to BC}(\rho^0_{A B}))\geq\delta$ for all recovery maps $\cR_{B \to B C}$ that satisfy~\eqref{eq_FRboundr}.
Since the set of trace-preserving completely positive maps from $B$ to $B \otimes C$ that satisfy~\eqref{eq_FRboundr} is compact\footnote{This set is bounded as the set of trace-preserving completely positive maps from $B$ to $B \otimes C$ is bounded (see Remark~\ref{rmk_TPCPcompact}). Furthermore, this set is closed since the set of trace-preserving completely positive maps from $B$ to $B \otimes C$ is closed (see Remark~\ref{rmk_TPCPcompact}) and the mapping $\cR_{B \to BC} \mapsto F(\rho_{ABC},\cR_{B \to BC}(\rho_{AB}))$ is continuous for all states $\rho_{ABC}$ (see Lemma~\ref{lem_contInR}). The Heine-Borel theorem then implies compactness.} and the function $f: \TPCP(B,B\otimes C) \ni \cR_{B \to BC} \mapsto 1- F(\rho^0_{ABC},\cR_{B\to BC}(\rho^0_{AB}) )  \in [0,1]$ is continuous (see Lemma~\ref{lem_contInR}), Weierstrass' theorem ensures that $\delta:=\min_{\cR_{B \to B C}} f( \cR_{B \to BC})$, where we optimize over the set of trace-preserving completely positive maps from $B$ to $B \otimes C$ that satisfy~\eqref{eq_FRboundr}, exists. By assumption, for every recovery map $\cR_{B \to BC}$ that satisfies~\eqref{eq_FRboundr} we have $f(\cR_{B \to BC})>0$ and hence $\delta \in (0,1]$.
If we insert any such recovery map $\cR_{B \to B C}$ into~\eqref{eq_FRboundr}, this gives
\begin{align} \label{eq_contradiction}
1-F\bigl(\rho_{A B C}, \cR_{B \to BC}(\rho_{A B})\bigr)+ \frac{\delta}{p} - \delta \leq  \frac{\ln(2)}{2} I(A : C | B)_{\rho}\ ,
\end{align}
which cannot be valid for sufficiently small $p$. More rigorously, assuming $1- F\bigl(\rho_{A B C}, \cR_{B \to BC}(\rho_{A B})\bigr) < \frac{\ln(2)}{2}I(A : C | B)_{\rho}$, as otherwise~\eqref{eq_contradiction} is violated for every $p \in (0,1]$,~\eqref{eq_contradiction} can be rewritten as
\begin{align}
p \geq \frac{\delta}{\frac{\ln(2)}{2}I(A : C | B)_{\rho} + \delta + F\bigl(\rho_{A B C},  \cR_{B \to BC}(\rho_{A B})\bigr) -1} =: \gamma \in (0,1]\ ,
\end{align}
since $C$ is assumed to be a finite-dimensional system and as such $I(A : C | B)_{\rho} < \infty$. Hence for $p < \gamma$, inequality~\eqref{eq_contradiction} is violated.
Since by~\cite{FR14} for any $p\in(0,1]$ there exists a recovery map $\cR_{B \to BC}$ of the form~\eqref{eq_mapForm} that satisfies~\eqref{eq_FRboundr} we conclude that for sufficiently small $p$ there exists a recovery map $\cR_{B \to BC}$ of the form~\eqref{eq_mapForm} that satisfies~\eqref{eq_FRboundr} and leaves $\rho_{ABC}^0$ invariant. However for recovery maps that leave $\rho_{ABC}^0$ invariant,~\eqref{eq_FRboundr} simplifies to~\eqref{eq_boundM} for all $p$.
Thus, there exists a recovery map $\cR_{B \to B C}$ of the form~\eqref{eq_mapForm} satisfying~\eqref{eq_boundM} that leaves $\rho^0_{A B C}$ invariant, i.e., $\cR_{B\to BC}(\rho^0_{A B}) = \rho^0_{A B C}$.
Since $\rho^0_{ABC}:= \rho_A \otimes \rho_{BC}$ is a Markov chain with marginal $\rho^{0}_{BC}=\rho_{BC}$, the condition $\cR_{B \to BC}(\rho^0_{AB})=\rho^0_{ABC}$ implies that $\cR_{B \to BC}(\rho_B) = \rho_{BC}$.

Have have thus shown that there exists a recovery map $\cR_{B \to BC}$ that satisfies~\eqref{eq_boundM} and fulfills
\begin{align} \label{eq_BtoBC}
\cR_{B \to BC} (\rho_B) = V_{BC} \rho_{BC}^{\frac{1}{2}}(U_B U_B^{\dagger} \otimes \id_C)  \rho_{BC}^{\frac{1}{2}} V_{BC}^\dagger  = \rho_{BC} \ .
\end{align}
Using the fact that $\cR_{B \to BC}$ is trace preserving and the invariance of the trace under unitaries we find
\begin{align}
1 = \tr\bigl(V_{BC} \rho_{BC}^{\frac{1}{2}}(U_B U_B^{\dagger} \otimes \id_C)  \rho_{BC}^{\frac{1}{2}} V_{BC}^\dagger \bigr) =\tr\bigl(\rho_{BC}(U_B U_B^{\dagger} \otimes \id_C)  \bigr) =\tr\bigl(\rho_{B} U_B U_B^{\dagger}  \bigr) \ .
\end{align}
This implies that $\tr\bigl(\rho_B (\id_B - U_B U_B^{\dagger}) \bigr) = 0$. Using the fact that $U_B U_B^{\dagger} \leq \id_B$, we conclude that $U_B U_B^\dagger = \id_B$ on the support of $\rho_B$. This simplifies~\eqref{eq_BtoBC} to $V_{BC} \rho_{BC} V_{BC}^\dagger = \rho_{BC}$, i.e., $V_{BC}$ and $\rho_{BC}$ commute which concludes the proof.
\end{proof}
Lemma~\ref{lem_removeUnitaries} implies that the mapping~\eqref{eq_mapForm} can be written as
\begin{align} \label{eq_mapDAV}
X_B \mapsto   \rho_{BC}^{\frac{1}{2}}W_{BC} (\rho_B^{-\frac{1}{2}} X_B  \rho_B^{-\frac{1}{2}} \otimes \id_C)  W_{BC}^{\dagger}\rho_{BC}^{\frac{1}{2}} \ ,
\end{align}
with $W_{BC} = V_{BC} U_{B}\otimes \id_C$ which is a unitary as $V_{BC}$ and $U_B$ are unitaries. Furthermore, $W_{BC}$ is such that~\eqref{eq_mapDAV} is trace-preserving.

\subsection*{\hypertarget{step_5_2}{Step 2: }Structure of a universal recovery map for fixed $A$ system}
In this step we show that the recovery map satisfying~\eqref{eq_prop_main} of the form~\eqref{eq_recoveryMapForm}, whose existence has been established in Step~\hyperlink{step_5_1}{1}, can be made universal without sacrificing the (partial) knowledge about its structure.
The idea is to apply Proposition~\ref{prop_finite} for the function family
\begin{align} 
\tilde \Delta_{\cR}(\rho): \, \,\, \,&\D(A\otimes B \otimes C)  \to \R \cup \{-\infty\} \nonumber \\
& \rho_{ABC} \mapsto F\bigl(\rho_{ABC},\cR_{B \to BC}(\rho_{AB}) \bigr) - 1 +\frac{\ln(2)}{2} I(A:C|B)_{\rho}\ .\label{eq_deltaLinearized}
\end{align}
We therefore need to verify that the assumptions of Proposition~\ref{prop_finite} are fulfilled. This is done by the following lemma. We first note that since $C$ is finite-dimensional this implies that $\tilde \Delta_{\cR}(\rho) < \infty$ for all $\rho \in \D(A \otimes B \otimes C)$.
\begin{lemma} \label{lem_linearizedProp}
Let $A$ be a separable and $B$ and $C$ finite-dimensional Hilbert spaces.
The function family $\tilde \Delta_{\cR}(\cdot)$ defined by~\eqref{eq_deltaLinearized} satisfies Properties~\ref{property_1}-\ref{property_5}.
\end{lemma}
\begin{proof}
We start by showing that $\tilde \Delta_{\cR}(\cdot)$ satisfies Property~\ref{property_1}.
   For $\rho^p_{\hat{A} A B C}$ as defined in~\eqref{eq_rhop}, we have
   \begin{align} \label{eq_ddecompose}
     F\bigl(\rho_{\hat A A B C}^p, \cR_{B \to BC}(\rho_{\hat A A B}^p)\bigr) = (1-p) F\bigl(\rho_{A B C}^0, \cR_{B \to BC}(\rho_{A B}^0)\bigr) + p F\bigl(\rho_{ABC}, \cR_{B \to BC}(\rho_{AB})\bigr) \ .
     \end{align}
     The density operator $ \cR_{B \to B C}(\rho^p_{\hat{A} A B})$ can be written as
    \begin{align}
   \cR_{B \to B C}(\rho^p_{\hat{A} A B}) = (1-p) \proj{0}_{\hat{A}} \otimes \cR_{B \to B C}(\rho^0_{A B}) + p \proj{1}_{\hat{A}} \otimes \cR_{B \to B C}(\rho_{A B}) \ .
   \end{align}
    The relevant density operators thus satisfy the orthogonality conditions for equality in Lemma~\ref{lem_fidelityconcave}, from which~\eqref{eq_ddecompose} follows.        
    Furthermore, as explained in the proof of Lemma~\ref{lem_deltaMeasRelProp1} we have 
    \begin{align} \label{eq_Idecompose}
      I(\hat A A :C|B)_{\rho^p} = (1-p) I(A:C|B)_{\rho^0} + p I(A:C|B)_{\rho} \ .
    \end{align}
    Equations~\eqref{eq_ddecompose} and~\eqref{eq_Idecompose} imply that
    \begin{align} \label{eq_Deltadecompose}
      \tilde \Delta_{\cR}(\rho^p) = (1-p) \tilde\Delta_{\cR}(\rho^0) + p \tilde\Delta_{\cR}(\rho) \ .
    \end{align}

We next verify that $\tilde\Delta_{\cR}(\cdot)$ fulfills Property~\ref{property_2}.
Let $\cR_{B \to BC}, \cR'_{B \to BC} \in \TPCP(B,B\otimes C)$, $\alpha \in [0,1]$ and $\bar \cR_{B \to BC} = \alpha \cR_{B \to BC} + (1-\alpha) \cR'_{B \to BC}$.
  Lemma~\ref{lem_fidelityconcave} implies that for any state $\rho_{ABC}$ on $A \otimes B \otimes C$
  \begin{multline}
     F\bigl(\rho_{ABC}, \bar{\cR}_{B \to B C}(\rho_{AB})\bigr) 
     = F\bigl(\rho_{ABC}, \alpha \cR_{B \to B C}(\rho_{AB}) + (1-\alpha) \cR'_{B \to B C}(\rho_{AB}) \bigl)\\
     \geq \alpha F\bigl(\rho_{ABC}, \cR_{B \to B C}(\rho_{AB})\bigr) + (1-\alpha) F\bigl(\rho_{ABC}, \cR'_{B \to B C}(\rho_{AB})\bigr) \ ,
  \end{multline}
  and, hence by the definition of $\tilde\Delta_{\cR}(\cdot)$
  \begin{align} \label{eq_Deltaconcave}
    \tilde\Delta_{\bar{\cR}}(\rho) \geq \alpha \tilde\Delta_{\cR}(\rho) + (1-\alpha) \tilde\Delta_{\cR'}(\rho) \ .
  \end{align}
  
The function $\rho \mapsto  \tilde\Delta_{\cR}(\rho)$ is continuous which clearly implies Property~\ref{property_4}. To see this, recall that by the Alicki-Fannes inequality $\rho \mapsto I(A:C|B)_{\rho}$ is continuous for a finite-dimensional $C$ system~\cite{AF03}. Furthermore, since $\rho_{AB} \mapsto \cR_{BC}(\rho_{AB})$ is continuous (see Lemma~\ref{lem_contInRho}), Lemma~B.9 of~\cite{FR14} implies that $\rho_{ABC} \mapsto F(\rho_{ABC},\cR_{B \to BC}(\rho_{AB}))$ is continuous, which then establishes Property~\ref{property_4}.
 
 Finally it remains to show that $\tilde\Delta_{\cR}(\cdot)$ satisfies Property~\ref{property_5}, which however follows directly by Lemma~\ref{lem_contInR}.

\end{proof}

Let $\cP \subseteq \TPCP(B,B \otimes C)$ be the convex hull of the set of trace-preserving completely positive mappings from the $B$ to the $B \otimes C$ system that are of the form~\eqref{eq_recoveryMapForm}. We note that the elements of $\cP$ are mappings of the form~\eqref{eq_unitalRecMap}, since a convex combination of unitary mappings are unital and a convex combination of trace-preserving maps remains trace-preserving. Proposition~\ref{prop_finite}, which is applicable as shown in Lemma~\ref{lem_linearizedProp} together with Step~\hyperlink{step_5_1}{1} therefore proves the assertion for a fixed $A$ system. 
\subsection*{\hypertarget{step_5_3}{Step 3: }Independence from the $A$ system}

Let $ \cS$ be the set of all density operators on $\bar A \otimes B \otimes C$ with a fixed marginal $\rho_{BC}$ on $B \otimes C$, where $B$ and $C$ are finite-dimensional Hilbert spaces and $\bar A$ is the infinite-dimensional Hilbert space $\ell^2$ of square summable sequences.

We note that the set of trace-preserving completely positive maps of the form~\eqref{eq_unitalRecMap} on finite-dimensional systems is compact, which follows by Remark~\ref{rmk_TPCPcompact} together with the fact that the intersection of a compact set and a closed set is compact. Hence, using Lemma~\ref{lem_linearizedProp} (in particular Properties~\ref{property_4} and~\ref{property_5}) and the result from Step~\hyperlink{step_5_2}{2} above, the same argument as in Step~\hyperlink{step_3_4}{4} of Section~\ref{sec_finiteDim} can be applied to conclude the existence of a recovery map $\cR_{B \to BC}$ of the form~\eqref{eq_unitalRecMap} such that $\tilde \Delta_{\cR}(\cS) \geq 0$. 

As every separable Hilbert space $A$ can isometrically embedded into $\bar A$~\cite[Theorem~II.7]{simon_book} and since $ \tilde \Delta_{\bar{\cR}}$ is invariant under isometries applied on the extension space $A$, we can conclude that the recovery map $\cR_{B \to BC}$ remains valid for any separable extension space $A$.
This proves the statement of Corollary~\ref{cor_main} for finite-dimensional $B$ and $C$ systems.

\section{Discussion} \label{sec_discussion}
Our main result is that for any density operator $\rho_{BC}$ on $B \otimes C$ there exists a recovery map $\cR_{B \to BC}$ such that the distance between \emph{any} extension $\rho_{ABC}$ of $\rho_{BC}$ acting on $A \otimes B \otimes C$ and $\cR_{B \to BC}(\rho_{AB})$ is bounded from above by the conditional mutual information $I(A:C|B)_{\rho}$. It is natural to ask whether such a map can be described as a simple and explicit function of $\rho_{BC}$. In fact, it was conjectured in~\cite{Kim13,BSW14} that~\eqref{eq_FR} holds for a very simple choice of map, namely
\begin{align} \label{eq_Petzmap}
 \cT_{B \to BC}\, : \, X_B \mapsto \rho_{B C}^{\frac{1}{2}} (\rho_B^{-\frac{1}{2}} X_B \rho_B^{-\frac{1}{2}} \otimes \id_C) \rho_{B C}^{\frac{1}{2}}  \ ,
\end{align}
called the \emph{transpose map} or \emph{Petz recovery map}. This conjecture, if correct, would have important consequences in obtaining remainder terms for the monotonicity of the relative entropy~\cite{BLW14}. As discussed in the introduction, if $\rho_{ABC}$ is such that it is a (perfect) quantum Markov chain or the $B$ system is classical, the claim of the conjecture is known to hold.  

One possible approach to prove a result of this form would be to start from the result~\eqref{eq_FR} for an unknown recovery map and then show that the transpose map $\cT_{B \to BC}$ cannot be much worse than any other recovery map. In fact, a theorem of Barnum and Knill~\cite{BK02} directly implies that when $\rho_{ABC}$ is pure, we have
\begin{align}\label{eq_BK}
F \big(\rho_{ABC}, \cT_{B \to BC}(\rho_{AB}) \big) \leq F(A;C|B)_{\rho} \leq \sqrt{F \big(\rho_{ABC}, \cT_{B \to BC}(\rho_{AB}) \big)} \ .
\end{align}
This shows that, if $\rho_{ABC}$ is pure, an inequality of the form~\eqref{eq_FR}, with the fidelity replaced by its square root, holds for the transpose map.
In order to generalize this to all states, one might hope that~\eqref{eq_BK} also holds for mixed states $\rho_{ABC}$. However, this turns out to be wrong even when the state $\rho_{ABC}$ is completely classical (see Appendix~\ref{app_transpose} for an example).

Another interesting question is whether the lower bound in terms of the measured relative entropy~\eqref{eq_main_new} can be improved to a relative entropy. Such an inequality is known to be false if we restrict the recovery map to be the transpose map~\eqref{eq_Petzmap}~\cite{WL12}, but it might be true when we optimize over all recovery maps. It is worth noting that in case such an inequality holds for any $\rho_{ABC}$ and a corresponding recovery map, then the argument presented in this work would imply that there exists a universal recovery map satisfying~\eqref{eq_main_new} with the relative entropy instead of the measured relative entropy. This can be seen by defining the function family $ \rho \mapsto \Delta_{\cR}(\rho):=I(A:C|B)_{\rho} - D(\rho_{ABC}\dd\cR_{B \to BC}(\rho_{AB}))$. Lemma~\ref{lem_relEnt}, the convexity of the relative entropy~\cite[Theorem~11.12]{nielsenChuang_book} and the lower semicontinuity of the relative entropy~\cite[Example~7.22]{holevo_book} imply that $\Delta_{\cR}(\cdot)$ satisfies Properties~\ref{property_1}-\ref{property_5}. As a result, Proposition~\ref{prop_finite} is applicable which can be used to prove the existence of a universal recovery map.

\appendix

\section*{Appendices}
\section{General facts about the fidelity}

The following lemma states a standard concavity property of the fidelity which is presented here for completeness and since we are interested in the case where equality holds.

\begin{lemma} \label{lem_fidelityconcave}
  For any density operators  $\rho$, $\rho'$, $\sigma$, and $\sigma'$, and for any $p \in [0,1]$ we have
  \begin{align} \label{eq_fidelityconcave}
    F\bigl(p \rho + (1-p) \rho', p \sigma + (1-p)\sigma' \bigr) \geq  p F(\rho, \sigma) + (1-p) F(\rho', \sigma') \ ,
  \end{align}
  with equality if both of $\rho$  and $\sigma$ are orthogonal to both of $\rho'$ and $\sigma'$. 
\end{lemma}

\begin{proof}
  Note first that for any two normalized and mutually orthogonal vectors $\ket{0}$ and $\ket{1}$ in an ancilla space, we have
  \begin{align}
   F\bigl(p \rho\! + \!(1\!-\!p) \rho', p \sigma\! +\! (1-p)\sigma' \bigr) 
   \geq     F\bigl(p \rho\! \otimes\! \proj{0} \!+\! (1\!-\!p) \rho'\! \otimes \!\proj{1}, p \sigma \!\otimes\! \proj{0} + (1-p)\sigma'\! \otimes\! \proj{1} \bigr) \ , 
   \end{align}
   because of the monotonicity of the fidelity under the partial trace. Furthermore,  if both of $\rho$ and $\sigma$ are orthogonal to both of $\rho'$ and $\sigma'$ then there exists a trace-preserving completely positive map that generates the corresponding state $\ket{0}$ or $\ket{1}$ of the ancilla system. This implies that, in this case, the inequality also holds in the other direction.   It therefore suffices to  prove~\eqref{eq_fidelityconcave} with $\rho$ and $\sigma$ replaced by $\rho \otimes \proj{0}$ and $\sigma \otimes \proj{0}$, and with $\rho'$ and $\sigma'$ replaced by $\rho' \otimes \proj{1}$ and $\sigma' \otimes \proj{1}$, respectively. In other words, it remains to show that, for the case where $\rho$ and $\sigma$ are orthogonal to $\rho'$ and $\sigma'$,~\eqref{eq_fidelityconcave} holds with equality, i.e., 
 \begin{align} \label{eq_fidelityequality}
 F(\bar{\rho}, \bar{\sigma})  = p F(\rho, \sigma) + (1-p) F(\rho', \sigma') \ ,
 \end{align}
 where $\bar{\rho} = p \rho + (1-p) \rho'$  and $\bar{\sigma} = p \sigma + (1-p) \sigma'$. 
   
 For this, let  $\ket{\phi}$,  $\ket{\phi'}$,  $\ket{\psi}$, and $\ket{\psi'}$ be purifications of  $\rho$, $\rho'$, $\sigma$, and $\sigma'$, respectively, such that $F(\rho, \sigma) = \spr{\phi}{\psi}$ and $F(\rho', \sigma') = \spr{\phi'}{\psi'}$.  It is easy to verify that
        \begin{align}
    \ket{\bar{\phi}} = \sqrt{p} \ket{\phi} \otimes \ket{0} + \sqrt{1-p} \ket{\phi'}  \otimes \ket{1} \quad \text{and} \quad  \ket{\bar{\psi}} = \sqrt{p} \ket{\psi}  \otimes \ket{0} + \sqrt{1-p} \ket{\psi'} \otimes \ket{1} 
  \end{align}
  are purifications of $\bar{\rho}$ and of $\bar{\sigma}$, respectively. Hence,  
  \begin{align}
    p F(\rho, \sigma) + (1-p) F(\rho', \sigma') 
     = p \spr{\phi}{\psi} + (1-p) \spr{\phi'}{\psi'} 
  = \spr{\bar{\phi}}{\bar{\psi}} 
  \leq F(\bar{\rho}, \bar{\sigma}) \ ,
  \end{align}
which proves one direction of~\eqref{eq_fidelityequality}. 
  
  To prove the other direction, let $\pi$ be  the projector onto the joint support of $\rho$ and $\sigma$, i.e., $\pi \rho = \rho$ and $\pi \sigma = \sigma$. Similarly, let $\pi'$ be the projector onto the joint support of $\rho'$ and $\sigma'$, i.e,. $\pi' \rho' = \rho'$ and $\pi' \sigma' = \sigma'$. By the condition that  $\rho$ and $\sigma$ are orthogonal to $\rho'$ and $\sigma'$,  the two projectors must be orthogonal, i.e., $\pi \pi' = 0$.  Furthermore, let $\ket{\bar{\phi}}$ be a purification of $\bar{\rho}$ and let $\ket{\bar{\psi}}$ be a purification of $\bar{\sigma}$ such that $F(\bar{\rho}, \bar{\sigma}) = \spr{\bar{\phi}}{\bar{\psi}}$.  Because 
 \begin{align}
   p \rho = \pi \bar{\rho} \pi \quad \text{and} \quad  (1-p) \rho' = \pi' \bar{\rho} \pi'
 \end{align}
  $\pi \ket{\bar{\phi}}$ and $\pi' \ket{\bar{\phi}}$ are purifications of $p \rho$ and $(1-p) \rho'$, respectively. Similarly, $\pi \ket{\bar{\psi}}$ and $\pi' \ket{\bar{\psi}}$ are purifications of $p \sigma$ and $(1-p) \sigma'$, respectively. Hence, we have 
  \begin{multline}
    F(\bar{\rho}, \bar{\sigma}) = \spr{\bar{\phi}}{\bar{\psi}} 
    = \bra{\bar{\phi}} \pi \ket{\bar{\psi}} + \bra{\bar{\phi}} \pi' \ket{\bar{\psi}} 
    \leq F\bigl(p \rho, p \sigma\bigr) + F\bigl((1-p) \rho', (1-p) \sigma'\bigr) \\
    = p F(\rho, \sigma) + (1-p) F(\rho', \sigma') \ .
  \end{multline}
  This proves the other direction of~\eqref{eq_fidelityequality} and thus concludes the proof. 
  \end{proof}

The following lemma generalizes the Fuchs-van de Graaf inequality which has been proven for states to non-negative operators. The result is standard and stated here for completeness.
\begin{lemma} \label{lem_tracedistancefidelity}
For any two non-negative operators $\rho$ and $\sigma$ with $\tr(\rho) \geq \tr(\sigma)$, the trace norm of their difference is bounded from above by
 \begin{align} \label{eq_tracedistancefidelity}
\norm{\rho- \sigma}_1 
 \leq 2\sqrt{\tr(\rho)^2-F(\rho, \sigma)^2} \ . 
\end{align}
\end{lemma}

\begin{proof}
 Let $\omega$ be a non-negative operator with $\tr(\omega) = \tr(\rho) - \tr(\sigma)$, whose support is orthogonal to the support of both $\rho$ and $\sigma$, and define $\sigma' = \sigma + \omega$. Then $\tr(\rho) = \tr(\sigma')$ and
   \begin{align}
  \norm{\rho- \sigma}_1 = \norm{\rho- \sigma'}_1  \qquad \text{and} \qquad  F(\rho, \sigma) = F(\rho, \sigma')  \ .
 \end{align}
 It therefore suffices to show that the claim holds for operators with $\tr(\rho) = \tr(\sigma) = c \in \R^+$. Furthermore for $c>0$, defining $\bar{\rho} = \rho / c$ and $\bar{\sigma} = \sigma / c$ and noting that
 \begin{align}
 \norm{\rho- \sigma}_1 = c  \norm{\bar{\rho}- \bar{\sigma}}_1 \qquad \text{and} \qquad F(\rho, \sigma) = c F(\bar{\rho}, \bar{\sigma}) \ ,
 \end{align}  
 it suffices to verify that the claim holds for $\tr(\rho) = \tr(\sigma) = 1$ which follows by the Fuchs-van de Graaf inequality~\cite{fuchs99}.
\end{proof}

\section{General facts about the measured relative entropy} \label{ap_measRelEnt}


\begin{definition} \label{def_MeasRelEnt}
The \emph{measured relative entropy} between density operators $\rho$ and $\sigma$ is defined as the supremum of the relative entropy with measured inputs over all POVMs $\cM = \{M_x\}$, i.e., 
\begin{align}
\MD(\rho \dd \sigma) = \sup \bigl\{ D( \cM(\rho) \dd \cM(\sigma) ) : \cM(\rho) = \sum_{x} \tr(\rho M_x) \proj{x} \text{ with } \sum_{x} M_x = \id \bigr\} \ ,
\end{align}
where $\{\ket{x}\}$ is a finite set of orthonormal vectors.
\end{definition}
This quantity was studied in~\cite{HP91,Hay01} where it was shown that $\frac{1}{n} \MD(\rho^{\otimes n} \dd \sigma^{\otimes n})$ converges to the relative entropy $D(\rho \dd \sigma):= \tr(\rho(\log \rho - \log \sigma))$.

\begin{lemma} \label{lem_relEnt}
Let $\rho$, $\rho'$, $\sigma$, and $\sigma'$ be density operators such that both $\rho$ and $\sigma$ are orthogonal to both $\rho'$ and $\sigma'$. For any $p\in [0,1]$ we have
\begin{align}
D\bigl(p \rho + (1-p) \rho' \, \dd \,p \sigma + (1-p) \sigma' \bigr) = p D(\rho \dd \sigma) + (1-p) D(\rho' \dd \sigma') \ .
\end{align}
\end{lemma}
\begin{proof}
By the orthogonality of $\rho$ and $\rho'$ (respectively $\sigma$ and $\sigma'$) we have
\begin{align}
\log\bigl(p \rho + (1-p) \rho' \bigr) = \log(p \rho) + \log \bigl((1-p) \rho' \bigr) = \log(p) + \log(1-p) + \log(\rho) + \log(\rho') 
\end{align}
and $\rho \log \rho' =  0$. Thus by definition of the relative entropy we obtain the desired statement.
\end{proof}

\begin{lemma}\label{lem_cqMeasRelEnt}
Let $\rho$, $\rho'$, $\sigma$, and $\sigma'$ be density operators such that both $\rho$ and $\sigma$ are orthogonal to both $\rho'$ and $\sigma'$. For any $p\in [0,1]$ we have
\begin{align}
\MD\bigl(p \rho + (1-p) \rho' \, \dd \,p \sigma + (1-p) \sigma' \bigr) = p \MD(\rho \dd \sigma) + (1-p) \MD(\rho' \dd \sigma') \ .
\end{align}
\end{lemma}
\begin{proof}
Let $\cM = \{M_x\}$, $\cM' = \{M'_y\}$ be measurements and define the POVM on $\cN$  whose elements are given by $\{ M_x \}_x \cup \{M'_y \}_{y}$. Then we can write
\begin{align}
\cN\bigl( p \rho + (1-p) \rho' \bigr) &= p \sum_{x} \tr(M_x \rho) \proj{x} + (1-p) \sum_{y} \tr(M'_y \rho') \proj{y} \ .
\end{align}
As a result using Lemma~\ref{lem_relEnt},
\begin{multline}
\MD\bigl(p \rho + (1-p) \rho' \, \dd \,p \sigma + (1-p) \sigma' \bigr) \geq D\Bigl( \cN\bigl(p \rho + (1-p) \rho' \bigr) \,\Big|\hspace{-0.6mm} \Big| \, \cN\bigl(p \sigma + (1-p) \sigma' \bigr) \Bigr)  \\
= p D\Bigl(\sum_{x} \tr(M_x \rho) \proj{x} \Big|\hspace{-0.6mm} \Big| \sum_{x} \tr(M_x \sigma) \proj{x}\Bigr)+(1-p) D\Bigl(\sum_{y} \tr(M'_y \rho') \proj{y} \Big|\hspace{-0.6mm} \Big| \sum_{y} \tr(M'_y \sigma') \proj{y}\Bigr) \ .
\end{multline}
As this inequality is valid for any measurements $\cM$ and $\cM'$, taking the supremum over such measurements gives
\begin{align}
\MD\bigl(p \rho + (1-p) \rho' \, \dd \,p \sigma + (1-p) \sigma' \bigr) \geq p \MD(\rho \dd \sigma) + (1-p) \MD(\rho' \dd \sigma') \ .
\end{align}

For the other direction, consider a measurement $\cM = \{M_x\}$. We can write
\begin{align}
\cM\bigl(p \rho + (1-p) \rho' \bigr) 
&= \sum_{x} p \,  \tr(M_x  \rho) \proj{x} + (1-p) \, \tr(M_x \rho') \proj{x} \ .
\end{align}
Combining this with the joint convexity of the relative entropy~\cite[Theorem~11.12]{nielsenChuang_book}, we get
\begin{align}
&\MD\bigl(p \rho + (1-p) \rho' \, \dd \,p \sigma + (1-p) \sigma' \bigr) = D\Bigl(\cM\bigl(p \rho + (1-p) \rho' \bigr) \Big|\hspace{-0.6mm} \Big| \cM\bigl(p \sigma + (1-p) \sigma'\bigr) \Bigr) \nonumber \\
&\leq p \, D\Bigl(\sum_{x} \tr( M_x \rho) \proj{x} \Big|\hspace{-0.6mm} \Big| \sum_{x} \tr( M_x \sigma) \proj{x} \Bigr)  + (1-p) \, D\Bigl(\sum_{x} \tr( M_x  \rho') \proj{x} \Big|\hspace{-0.6mm} \Big| \sum_{x} \tr( M_x \sigma') \proj{x}\Bigr) \nonumber\\
&\leq p \MD(\rho \dd \sigma)+ (1-p) \MD( \rho'\dd \sigma')  \ .
\end{align}
\end{proof}

\begin{lemma} \label{lem_convexityMeasRelEnt}
For density operators $\rho$, $\sigma$, and $\sigma'$ and $p\in [0,1]$ the measured relative entropy satisfies 
\begin{align}
\MD\bigl(\rho \dd p \sigma+ (1-p) \sigma'   \bigr) \leq p \, \MD(\rho \dd \sigma)+  (1-p) \, \MD(\rho \dd \sigma')  \ .
\end{align}
\end{lemma}
\begin{proof}
For any measurement $\cM$,
\begin{align}
D\bigl(\cM(\rho) \dd \cM(  p \sigma + (1-p) \sigma') \bigr)
&= D\bigl(\cM(\rho) \dd \, p \, \cM(\sigma) + (1-p) \cM( \sigma') \bigr)\nonumber \\
&\leq  p \, D\bigl(\cM(\rho) \dd \cM(\sigma)\bigr) + (1-p) \, D\bigl(\cM(\rho) \dd \cM(\sigma')\bigr)\nonumber  \\
&\leq p \MD(\rho \dd \sigma) + (1-p) \MD(\rho \dd \sigma')  \ , 
\end{align}
where the first inequality step uses the convexity of the relative entropy~\cite[Theorem~11.12]{nielsenChuang_book}.
Taking the supremum over $\cM$, we get the desired result.
\end{proof}


\section{Basic topological facts} 
For completeness we state here some standard topological facts about density operators and trace-preserving completely positive maps.
\begin{lemma}\label{lem_DensityOpCompact}
Let $\alpha \in \R^+$. The space of non-negative operators on a finite-dimensional Hilbert space $E$ with trace smaller or equal to $\alpha$ (respectively equal to $\alpha$) is compact.
\end{lemma}
\begin{proof} 
Let $\D'(E):=\{\rho \in \Pos(E):\tr(\rho)\leq\alpha \}$ denote the set non-negative operators on $E$ with trace not larger than one, where $\Pos(E)$ is the set of non-negative operators on $E$. Consider the ball $\cB:=\{e \in E: \norm{e}\leq\alpha\}$ which is compact. The function $\cB \ni e \mapsto f(e)=e e^{\dagger} \in \D'(E)$ is continuous and thus the set $f(\cB)=\{e e^\dagger : e \in E, \norm{e}\leq\alpha\}$ is compact, as continuous functions map compact sets to compact sets. By the spectral theorem it follows that $\D'(E)= \mathrm{conv} f(\cB)$. As the convex hull of every compact set is compact this proves the assertion. The same argumentation (by replacing the inequalities with equalities) proves that the set of non-negative operators on $E$ with trace $\alpha$ is compact.
\end{proof}

\begin{lemma}\label{lem_DensityOpCompactFixedMarginal}
Let $E$, $G$ be finite-dimensional Hilbert spaces and let $\sigma_G \in \Pos(G)$. The space of non-negative operators on $E \otimes G$ with a marginal on $G$ smaller or equal to $\sigma_G$ (respectively equal to $\sigma_G$) is compact.
\end{lemma}
\begin{proof}
Let $\sigma_G \in \Pos(G)$. By Lemma~\ref{lem_DensityOpCompact}, the set of non-negative operators on $E \otimes G$ with trace not larger than $\alpha \in \R^+$ is compact.
The set $\{X \in E \otimes G: \tr_{E}(X)\leq\rho_{G} \}$ is closed. The intersection of a compact set and a closed set is compact which implies that $\{ X \in \Pos(E \otimes G ): \tr_{E}(X)\leq\rho_{G} \}$ is compact. Since the set $\{X \in E \otimes G: \tr_{E}(X)=\rho_{G} \}$ is closed the same argumentation shows that $\{ X \in \Pos(E \otimes G ): \tr_{E}(X)=\rho_{G} \}$ is compact.
\end{proof}

\begin{remark}  \label{rmk_TPCPcompact}
Let $E$ and $G$ be two finite-dimensional Hilbert spaces. The space of trace-non-increasing (respectively trace-preserving) completely positive maps from $E$ to $G$ is compact. To see this note that Lemma~\ref{lem_DensityOpCompactFixedMarginal} implies that the set $\cF := \{ X \in \Pos(E \otimes G ): \tr_{G}(X)\leq\id_{E} \}$ is compact. By the Choi-Jamiolkowski representation $\cF$ is however isomorphic to the set of all trace-non-increasing completely positive maps from $E$ to $G$. The same argumentation applied to the set $\cF := \{ X \in \Pos(E \otimes G ): \tr_{G}(X)=\id_{E} \}$ shows that the set of trace-preserving completely positive maps from $E$ to $G$ is compact.
\end{remark}

\begin{lemma} \label{lem_contInR}
Let $G$ and $K$ be finite-dimensional Hilbert spaces and let $\sigma_{EGK} \in \D(E \otimes G \otimes K)$. The mapping $\TPCP(G,G\otimes K) \ni \cR \mapsto F(\sigma_{EGK}, \cR_{G \to GK}(\sigma_{EGK})) \in [0,1]$ is continuous.
\end{lemma}
\begin{proof}
This follows directly from the continuity of $\cR \mapsto \cR_{G \to GK}(\sigma_{EG})$ and the continuity of the fidelity (see, e.g., Lemma~B.9 of~\cite{FR14}).
\end{proof}

\begin{lemma} \label{lem_contInRho}
Let $E$, $G$, and $K$ be separable Hilbert spaces and $\cR \in \TPCP(G, K) $. Then the mapping $\D(E \otimes G) \ni~X \mapsto \cI_{E} \otimes \cR_{G \to K}(X_{EG})\in \D(E \otimes K)$ is continuous.
\end{lemma}
\begin{proof}
As the map is linear it suffices to show that it is bounded. For that we can decompose $X = P - N$ with $P$ and $N$ orthogonal non-negative operators. Then we have 
\begin{align}
\| \cI_{E} \otimes \cR_{G \to K}(X) \|_1 \leq \| \cI_{E} \otimes \cR_{G \to K}(P) \|_1 + \| \cI_{E} \otimes \cR_{G \to K}(N) \|_1 
= \tr(P) + \tr(N) 
= \| X \|_{1} \ .
\end{align}
\end{proof}

\section{Touching sets lemma}
We prove here a basic fact that is used in the proof of Theorem~\ref{thm_main}.
\begin{lemma} \label{lem_meanVal}
Let $K_0$ and $K_1$ be two sets such that $K_0 \cup K_1 = [0,1]$ and $0 \in K_0$, $1 \in K_1$. Then for any $\delta > 0$ there exists $u\in K_0$ and $v\in K_1$ such that $0 \leq v-u \leq \delta$.
\end{lemma}
\begin{proof}
We define $\mu := \inf K_1$ and distinguish between the two cases $\mu \in K_0$ and $\mu \not \in K_0$.

If $\mu \in K_0$, it suffices to show that for any $\delta > 0$ we have $[\mu, \mu+\delta] \cap K_1 \neq \emptyset$, since by choosing $u=\mu$ this implies that $u \in K_0$ and that there exists a $v \in  [\mu, \mu+\delta]$ such that $v \in K_1$. By contradiction, we assume that $[\mu, \mu+\delta] \cap K_1 = \emptyset$. This implies that either $\inf K_1 < \mu$ or $\inf K_1 \geq \mu + \delta$, which contradicts $\mu := \inf K_1$.

If $\mu \not \in K_0$ it suffices to show that for any $\delta >0$ we have $[\mu - \delta, \mu] \cap K_0 \neq \emptyset$, since by choosing $v=\mu$ this ensures that $ v \in K_1$ and that there exists a $u \in [\mu - \delta, \mu]$ such that $u \in K_0$. Assume by contradiction that $[\mu - \delta, \mu] \cap K_0 = \emptyset$, which implies that $[\mu - \delta, \mu] \subset K_1$. This however contradicts $\mu := \inf K_1$.

\end{proof}

\section{Properties of projected states} 
We first prove variant of the \emph{gentle measurement lemma}~\cite{winter99}, which is used repeatedly in the proof of Theorem~\ref{thm_main}. 
\begin{lemma} \label{lem_varGentleMeas}
Let $E$ and $G$ be separable Hilbert spaces and let $\Pi_G$ be a finite-rank projector on $G$. For any non-negative operator $\sigma_{EG}$ on $E \otimes G$ we have
\begin{align}
F\!\left( \sigma_{EG}, \frac{(\id_E \otimes \Pi_G) \sigma_{EG} (\id_E \otimes \Pi_G)}{\tr\bigl((\id_E \otimes \Pi_G) \sigma_{EG}\bigr)} \right)^2 \geq \tr(\Pi_G \sigma_{EG}) 
\end{align}
and
\begin{align}
F\bigl( \sigma_{EG}, (\id_E \otimes \Pi_G) \sigma_{EG} (\id_E \otimes \Pi_G) \bigr) \geq  \tr(\Pi_G \sigma_{EG}) \ .
\end{align}

\end{lemma}
\begin{proof}
Let $\ket{\psi}$ be a purification of $\sigma_{EG}$ then by Uhlmann's theorem~\cite{Uhl76} we find
\begin{align}
F\!\left( \sigma_{EG}, \frac{(\id_E \otimes \Pi_G) \sigma_{EG} (\id_E \otimes \Pi_G)}{\tr\bigl((\id_E \otimes \Pi_G) \sigma_{EG}\bigr)} \right)^2 \geq \frac{(\bra{\psi} \Pi_G \ket{\psi})^2}{\tr\bigl((\id_E \otimes \Pi_G) \sigma_{EG}\bigr)} = \tr(\Pi_G \sigma_{EG}) 
\end{align}
and
\begin{align}
F\bigl( \sigma_{EG}, (\id_E \otimes \Pi_G) \sigma_{EG} (\id_E \otimes \Pi_G) \bigr)^2 \geq (\bra{\psi} \Pi_G \ket{\psi})^2= \tr(\Pi_G \sigma_{EG})^2 \ .
\end{align}
\end{proof}

We next prove a basic statement about converging projectors that is used several times in the proof of Theorem~\ref{thm_main}.
\begin{lemma} \label{lem_convergenceWeak}
Let $E$ be a separable Hilbert space and let $\{ \Pi_E^e\}_{e \in E}$ be a sequence of finite-rank projectors on $E$ which converges to $\id_E$ with respect to the weak operator topology. Then for any density operator $\sigma_E$ on $E$ we have $\lim_{e \to \infty}\tr(\Pi_E^e \sigma_E )=\tr(\sigma_E)$.
\end{lemma}
\begin{proof}
By assumption the Hilbert space $E$ is separable which implies that any state $\sigma_{E}$ can be written as $\sigma_{E}=\sum_i p_i \proj{x_i}$, where $p_i \geq 0$, $\sum_i p_i =1$ and $\{\ket{x_i} \}_i$ is an orthonormal basis on $E$. As the sequence $\{\Pi_E^e \}_{e \in \N}$ weakly converges to $\id_E$, we find
\begin{align} \label{eq_convergenceSeqA}
\lim_{e \to \infty} \tr (\Pi_E^e  \sigma_{E} ) 
= \lim_{e \to \infty} \sum_i p_i \bra{x_i} \Pi_E^e\ket{x_i}
= \sum_i  p_i  \lim_{e \to \infty}\bra{x_i} \Pi_E^e\ket{x_i}  
= \sum_i p_i \bra{x_i}  \id_{E} \ket{x_i} =\tr(\sigma_E) \ ,
\end{align}
where the second step uses dominated convergence that is applicable since $|\bra{x_i} \Pi_E^e\ket{x_i}| \leq |\bra{x_i} \id_E \ket{x_i}|$ for all $e \in \N$. 
\end{proof}

Let $E$ and $G$ be separable Hilbert spaces and let $\cS$ denote the set of bipartite density operators on $E \otimes G$ with a fixed marginal $\sigma_{G}$ on $G$. Let $\{\Pi_E^e \}_{e \in \N}$ be a sequence of projectors with rank $e$ that weakly converge to $\id_E$ and $\cS^e$ be the set of bipartite states on $E \otimes G$ whose marginal on $E$ is contained in the support of $\Pi^e_E$ and whose marginal on $G$ is identical to $\sigma_{G}$.

\begin{lemma} \label{lem_projSequence}
For every $\sigma_{EG} \in \cS$ there exists a sequence $\{\sigma_{EG}^e \}_{e \in \N}$ with $\sigma_{EG}^e \in \cS^e$ that converges to $\sigma_{EG}$ with respect to the trace norm.
\end{lemma}
\begin{proof}
For $\sigma_{EG} \in \cS$, let
\begin{align}
\bar \sigma_{EG}^e := \frac{(\Pi_E^e \otimes \id_{G})\sigma_{EG}(\Pi_E^e \otimes \id_{G})}{\tr\bigl( (\Pi_E^e \otimes \id_{G}) \sigma_{EG} \bigr)} \ ,
\end{align}
which has the desired support on $E$, however, $\bar \sigma_{G}^e \ne \sigma_{G}$ in general. This is fixed by considering 
\begin{align}
\sigma_{EG}^e := \tr\bigl( (\Pi_E^e \otimes \id_{G}) \sigma_{EG} \bigr) \bar \sigma_{EG}^e + \proj{0}_E \otimes\tr_E\bigl((\Pi_E^{e\perp} \otimes \id_{G})\sigma_{EG}(\Pi_E^{e\perp} \otimes \id_{G})\bigr)_{G} \ ,
\end{align}
where $\ket{0}_E$ is a normalized state on $E$.
Since the partial trace on $E$ is cyclic on $E$ we obtain
\begin{multline}
\sigma_{G}^e = \tr_E(\sigma_{EG}^e) 
= \tr_E\bigl((\Pi_E^e \otimes \id_{G})\sigma_{EG}(\Pi_E^e \otimes \id_{G})\bigr)+\tr_E\bigl((\Pi_E^{e\perp} \otimes \id_{G})\sigma_{EG}(\Pi_E^{e\perp} \otimes \id_{G})\bigr) \\
= \tr_E\bigl((\Pi_E^e \otimes \id_{G})\sigma_{EG}\bigr)+\tr_E\bigl((\Pi_E^{e\perp} \otimes \id_{G})\sigma_{EG}\bigr)
=\tr_E(\sigma_{EG}) = \sigma_{G} \ .
\end{multline}
By the multiplicativity of the trace norm under tensor products and since $\norm{A}_1=\tr(\sqrt{A^\dagger A})$, the triangle inequality implies that
\begin{multline}
\norm{\bar \sigma_{EG}^e - \sigma_{EG}^e}_1
\leq 1-  \tr\bigl( (\Pi_E^e \otimes \id_{G}) \sigma_{EG} \bigr) + \norm{\tr_E\bigl((\Pi_E^{e\perp} \otimes \id_{G})\sigma_{EG}(\Pi_E^{e\perp} \otimes \id_{G})\bigr)}_1 \\
= 1-  \tr\bigl( (\Pi_E^e \otimes \id_{G}) \sigma_{EG} \bigr) + \tr\bigl((\Pi_E^{e\perp} \otimes \id_{G})\sigma_{EG}\bigr)
= 2 \bigl(1-  \tr (\Pi_E^e  \sigma_{E} ) \bigr) \ . \label{eq_convergenceSeqAA}
\end{multline}
Lemma~\ref{lem_convergenceWeak} now implies that $\lim_{e \to \infty} \tr (\Pi_E^e  \sigma_{E} )  = 1$.
We note that the sequence $\{\bar \sigma_{EG}^e\}_{e \in \N}$ converges to $\sigma_{EG}$ in the trace norm since by the Fuchs-van de Graaf inequality~\cite{fuchs99}, Lemma~\ref{lem_varGentleMeas} and Lemma~\ref{lem_convergenceWeak}
\begin{align}
\lim_{e \to \infty} \norm{\sigma_{EG} - \bar \sigma_{EG}^e}_1 \leq \lim_{e \to \infty} 2\sqrt{1-F(\sigma_{EG},\bar \sigma_{EG}^e)^2} \leq 
 \lim_{e \to \infty} 2\sqrt{1-\tr(\Pi_E^e \sigma_{E})} =0 \ .
\end{align}
Combining this with~\eqref{eq_convergenceSeqAA} and the triangle inequality proves that $\{\sigma_{EG}^e \}_{e\in \N}$ converges to $\sigma_{EG}$ in the trace norm. 
\end{proof}

\section{The transpose map is not square-root optimal}
\label{app_transpose}
As discussed in Section~\ref{sec_discussion}, for pure states $\rho_{ABC}$ it is known~\cite{BK02} that
\begin{align}\label{eq_BKApp}
 F(A;C|B)_{\rho} \leq \sqrt{F \big(\rho_{ABC}, \cT_{B \to BC}(\rho_{AB}) \big)} \ 
\end{align}
holds for $\cT_{B \to BC}$ the transpose map. In this appendix we show that~\eqref{eq_BKApp} does not hold for all mixed states.
Let $\dim A = \dim B = \dim C = 2$ and consider the state
\begin{align}
\rho_{ABC} = \frac{1}{2} \proj{0}_{A} \otimes \proj{0}_{B} \otimes \proj{0}_C + \frac{1}{8} \proj{1}_A \otimes \id_{BC} \ .
\end{align}
The transpose map satisfies
\begin{align}
\cT_{B \to BC}(\proj{0}_B) = \frac{5}{6} \proj{00}_{BC} + \frac{1}{6} \proj{01}_{BC} \quad \text{and} \quad \cT_{B \to BC}(\proj{1}_B) = \frac{1}{2} \proj{10}_{BC} + \frac{1}{2} \proj{11}_{BC} \ .
\end{align}
If we consider a recovery map $\cR_{B \to BC}$ that is defined by
\begin{align}
\cR_{B \to BC}(\proj{0}_B) = \proj{00}_{BC} \quad \text{ and } \quad \cR_{B \to BC}(\proj{1}_B) = \frac{1}{3} \left(\proj{01}_{BC} + \proj{10}_{BC} + \proj{11}_{BC}\right) \ ,
\end{align}
we find $F(\rho_{ABC}, \cR_{B \to BC}(\rho_{AB})) > 0.9829$ and $\sqrt{F(\rho_{ABC}, \cT_{B \to BC}(\rho_{AB}))} < 0.9696$, which shows that \eqref{eq_BKApp} cannot hold since $F\bigl(\rho_{ABC}, \cR_{B \to BC}(\rho_{AB})\bigr) \leq F(A;C|B)_{\rho}$.

%
%
This does not show that one cannot prove a non-trivial guarantee on the performance of the transpose map relative to the optimal recovery map, but it suggests that such a guarantee would have to be worse than the square root (and actually worse that the fourth root as well using another example), or perhaps it is more naturally expressed using a different distance measure (using similar examples, the trace distance does not seem to be a good candidate, either).
We further note that this example does not show that~\eqref{eq_FR} is wrong for the transpose map.

\section*{Acknowledgments}
We thank Mario Berta, Marco Tomamichel, and Volkher Scholz for making us aware of the application of our result to topological order of quantum states (described in Section~\ref{sec_applications}).
We also thank Mario Berta, Fernando Brand{\~a}o, Philipp Kammerlander, Joseph Renes, Burak \c{S}ahino\u{g}lu, Volkher Scholz, Marco Tomamichel, Michael Walter and Mark Wilde for discussions about approximate Markov chains. This project was supported by the European Research Council (ERC) via grant No.~258932, by the Swiss National Science Foundation (SNSF) via the National Centre of Competence in Research ``QSIT'', and by the European Commission via the project ``RAQUEL''.

  \bibliographystyle{abbrv}
  
  \bibliography{bibliofile}

\begin{thebibliography}{10}

\bibitem{AF03}
R.~Alicki and M.~Fannes.
\newblock Continuity of quantum conditional information.
\newblock {\em Journal of Physics A: Mathematical and General}, 37(5):55--57,
  2004.

\bibitem{BK02}
H.~Barnum and E.~Knill.
\newblock Reversing quantum dynamics with near-optimal quantum and classical
  fidelity.
\newblock {\em Journal of Mathematical Physics}, 43(5):2097--2106, 2002.

\bibitem{BLW14}
M.~Berta, M.~Lemm, and M.~M. Wilde.
\newblock Monotonicity of quantum relative entropy and recoverability, 2014.
\newblock \href{http://arxiv.org/abs/1412.4067}{arXiv:1412.4067}.

\bibitem{BSW14}
M.~Berta, K.~P. Seshadreesan, and M.~M. Wilde.
\newblock R\'enyi generalizations of the conditional quantum mutual
  information.
\newblock {\em Journal of Mathematical Physics}, 56(2), 2015.

\bibitem{TB15}
M.~Berta and M.~Tomamichel.
\newblock The fidelity of recovery is multiplicative, 2015.
\newblock \href{http://arxiv.org/abs/1502.07973}{arXiv:1502.07973}.

\bibitem{bourbaki_book}
N.~Bourbaki.
\newblock {\em Elements of Mathematics: General Topology}.
\newblock Hermann, \'{E}diteures des {S}ciences et des {A}rts, 1966.

\bibitem{BraHar13_1}
F.~G. Brand\~{a}o and A.~W. Harrow.
\newblock Product-state approximations to quantum ground states.
\newblock In {\em Proceedings of the Forty-fifth Annual ACM Symposium on Theory
  of Computing}, STOC '13, pages 871--880, New York, NY, USA, 2013. ACM.

\bibitem{BraHar13_2}
F.~G. Brand\~{a}o and A.~W. Harrow.
\newblock Quantum de {F}inetti theorems under local measurements with
  applications.
\newblock In {\em Proceedings of the Forty-fifth Annual ACM Symposium on Theory
  of Computing}, STOC '13, pages 861--870, New York, NY, USA, 2013. ACM.

\bibitem{BHOS14}
F.~G.~S.~L. Brand\~ao, A.~W. Harrow, J.~Oppenheim, and S.~Strelchuk.
\newblock Quantum conditional mutual information, reconstructed states, and
  state redistribution.
\newblock {\em Physical Review Letters}, 115(5):050501, July 2015.
\newblock \href{http://www.arxiv.org/abs/1411.4921}{arXiv:1411.4921}.

\bibitem{BHV06}
S.~Bravyi, M.~B. Hastings, and F.~Verstraete.
\newblock Lieb-{R}obinson bounds and the generation of correlations and
  topological quantum order.
\newblock {\em Phys. Rev. Lett.}, 97:050401, Jul 2006.

\bibitem{CSW12}
M.~Christandl, N.~Schuch, and A.~Winter.
\newblock Entanglement of the antisymmetric state.
\newblock {\em Communications in Mathematical Physics}, 311(2):397--422, 2012.

\bibitem{DW15}
N.~Datta and M.~M. Wilde.
\newblock Quantum {M}arkov chains, sufficiency of quantum channels, and
  {R}\'enyi information measures, 2015.
\newblock \href{http://arxiv.org/abs/1501.05636}{arXiv:1501.05636}.

\bibitem{einsiedler_book}
M.~Einsiedler and T.~Ward.
\newblock {\em Ergodic Theory}.
\newblock Springer, 2010.

\bibitem{FR14}
O.~Fawzi and R.~Renner.
\newblock Quantum conditional mutual information and approximate {M}arkov
  chains.
\newblock {\em Communications in Mathematical Physics}, 340(2):575--611, 2015.

\bibitem{fuchs99}
C.~Fuchs and J.~van~de Graaf.
\newblock Cryptographic distinguishability measures for quantum-mechanical
  states.
\newblock {\em IEEE Transactions on Information Theory}, 45(4):1216 --1227, May
  1999.

\bibitem{Fuc96}
C.~A. Fuchs.
\newblock Distinguishability and accessible information in quantum theory.
\newblock {\em PhD Thesis, University of New Mexico}, 1996.
\newblock \href{http://arxiv.org/abs/quant-ph/9601020}{arXiv:quant-ph/9601020}.

\bibitem{FAR11}
F.~Furrer, J.~{\AA}berg, and R.~Renner.
\newblock Min- and max-entropy in infinite dimensions.
\newblock {\em Communications in Mathematical Physics}, 306(1):165--186, 2011.

\bibitem{Hay01}
M.~Hayashi.
\newblock Asymptotics of quantum relative entropy from a representation
  theoretical viewpoint.
\newblock {\em Journal of Physics A: Mathematical and General}, 34(16):3413,
  2001.

\bibitem{HJPW04}
P.~Hayden, R.~Jozsa, D.~Petz, and A.~Winter.
\newblock Structure of states which satisfy strong subadditivity of quantum
  entropy with equality.
\newblock {\em Communications in Mathematical Physics}, 246(2):359--374, 2004.

\bibitem{HP91}
F.~Hiai and D.~Petz.
\newblock The proper formula for relative entropy and its asymptotics in
  quantum probability.
\newblock {\em Communications in Mathematical Physics}, 143(1):99--114, 1991.

\bibitem{holevo_book}
A.~S. Holevo.
\newblock {\em Quantum Systems, Channels, Information}.
\newblock De Gruyter Studies in Mathematical Physics 16, 2012.

\bibitem{ILW08}
B.~Ibinson, N.~Linden, and A.~Winter.
\newblock Robustness of quantum {M}arkov chains.
\newblock {\em Communications in Mathematical Physics}, 277(2):289--304, 2008.

\bibitem{JRS03}
R.~Jain, J.~Radhakrishnan, and P.~Sen.
\newblock A lower bound for the bounded round quantum communication complexity
  of set disjointness.
\newblock In {\em IEEE 44th Annual IEEE Symposium on Foundations of Computer
  Science}, pages 220--229, Oct 2003.

\bibitem{KLLRX12}
I.~Kerenidis, S.~Laplante, V.~Lerays, J.~Roland, and D.~Xiao.
\newblock Lower bounds on information complexity via zero-communication
  protocols and applications.
\newblock In {\em IEEE 53rd Annual Symposium on Foundations of Computer
  Science}, pages 500--509, Oct 2012.

\bibitem{Kim13}
I.~Kim.
\newblock Application of conditional independence to gapped quantum many-body
  systems, 2013.
\newblock
  \href{http://www.physics.usyd.edu.au/quantum/Coogee2013/Presentations/Kim.pdf}{\url{http://www.physics.usyd.edu.au/quantum/Coogee2013/Presentations/Kim.pdf}}.

\bibitem{kim_phd}
I.~H. Kim.
\newblock Conditional independence in quantum many-body systems.
\newblock {\em PhD thesis, Caltech}, 2013.

\bibitem{KitPres06}
A.~Kitaev and J.~Preskill.
\newblock Topological entanglement entropy.
\newblock {\em Phys. Rev. Lett.}, 96:110404, Mar 2006.

\bibitem{LW06}
M.~Levin and X.-G. Wen.
\newblock Detecting topological order in a ground state wave function.
\newblock {\em Phys. Rev. Lett.}, 96:110405, Mar 2006.

\bibitem{LiWin14}
K.~Li and A.~Winter.
\newblock Squashed entanglement, k-extendibility, quantum {M}arkov chains, and
  recovery maps, 2014.
\newblock \href{http://arxiv.org/abs/1410.4184}{arXiv:1410.4184}.

\bibitem{LieRus73}
E.~H. Lieb and M.~B. Ruskai.
\newblock Proof of the strong subadditivity of quantum-mechanical entropy.
\newblock {\em Journal of Mathematical Physics}, 14(12):1938--1941, 1973.

\bibitem{MLDSFT13}
M.~M\"{u}ller-Lennert, F.~Dupuis, O.~Szehr, S.~Fehr, and M.~Tomamichel.
\newblock On quantum {R}\'enyi entropies: A new generalization and some
  properties.
\newblock {\em Journal of Mathematical Physics}, 54(12):--, 2013.

\bibitem{nielsenChuang_book}
M.~A. Nielsen and I.~L. Chuang.
\newblock {\em Quantum Computation and Quantum Information}.
\newblock Cambridge University Press, 2000.

\bibitem{Pet86}
D.~Petz.
\newblock Sufficient subalgebras and the relative entropy of states of a von
  {N}eumann algebra.
\newblock {\em Communications in Mathematical Physics}, 105(1):123--131, 1986.

\bibitem{Pet03}
D.~Petz.
\newblock Monotonicity of quantum relative entropy revisited.
\newblock {\em Reviews in Mathematical Physics}, 15(01):79--91, 2003.

\bibitem{simon_book}
M.~Reed and B.~Simon.
\newblock {\em Functional Analysis}.
\newblock Elsevier, Academic Press, 1980.

\bibitem{SW14}
K.~P. Seshadreesan and M.~M. Wilde.
\newblock Fidelity of recovery, geometric squashed entanglement, and
  measurement recoverability, 2014.
\newblock \href{http://arxiv.org/abs/1410.1441}{arXiv:1410.1441}.

\bibitem{tomamichel_phd}
M.~Tomamichel.
\newblock A framework for non-asymptotic quantum information theory.
\newblock {\em PhD thesis, ETH Zurich}, 2012.
\newblock \href{http://arxiv.org/abs/1203.2142}{arXiv:1203.2142}.

\bibitem{TCR10}
M.~Tomamichel, R.~Colbeck, and R.~Renner.
\newblock Duality between smooth min- and max-entropies.
\newblock {\em IEEE Transactions on Information Theory}, 56(9):4674--4681, Sept
  2010.

\bibitem{Touch14}
D.~Touchette.
\newblock Quantum information complexity and amortized communication, 2014.
\newblock \href{http://arxiv.org/abs/1404.3733}{arXiv:1404.3733}.

\bibitem{Uhl76}
A.~Uhlmann.
\newblock The ``transition probability'' in the state space of a {*}-algebra.
\newblock {\em Reports on Mathematical Physics}, 9(2):273 -- 279, 1976.

\bibitem{wilde_book}
M.~M. Wilde.
\newblock {\em Quantum Information Theory}.
\newblock Cambridge University Press, June 2013.

\bibitem{wilde15}
M.~M. Wilde.
\newblock Recoverability in quantum information theory, 2015.
\newblock \href{http://arxiv.org/abs/1505.04661}{arXiv:1505.04661}, accepted in
  \textit{Proceedings of the Royal Society~A}.

\bibitem{WWY14}
M.~M. Wilde, A.~Winter, and D.~Yang.
\newblock Strong converse for the classical capacity of entanglement-breaking
  and {H}adamard channels via a sandwiched {R}\'enyi relative entropy.
\newblock {\em Communications in Mathematical Physics}, 331(2):593--622, 2014.

\bibitem{winter99}
A.~Winter.
\newblock Coding theorem and strong converse for quantum channels.
\newblock {\em IEEE Transactions on Information Theory}, 45(7):2481--2485, Nov
  1999.

\bibitem{WL12}
A.~Winter and K.~Li.
\newblock A stronger subadditivity relation? {W}ith applications to squashed
  entanglement, sharability and separability, 2012.
\newblock
  \href{http://www.maths.bris.ac.uk/~csajw/stronger_subadditivity.pdf}{\url{http://www.maths.bris.ac.uk/~csajw/stronger_subadditivity.pdf}}.

\bibitem{Zha12}
L.~Zhang.
\newblock Conditional mutual information and commutator.
\newblock {\em International Journal of Theoretical Physics}, 52(6):2112--2117,
  2013.

\end{thebibliography}
  
  \end{document}